\documentclass[final,3p,times]{elsarticle}
\usepackage{CJK}
\usepackage{graphicx} 
\usepackage{amssymb}
\usepackage{amsmath}
\usepackage{subfigure}
\usepackage{graphicx}
\usepackage{epstopdf}
\usepackage{array}
\usepackage{bibentry}
\usepackage{mathrsfs}
\usepackage{color}
\usepackage{multirow}
\biboptions{sort&compress}
\usepackage{breqn}
\usepackage{lineno,hyperref}
\DeclareMathOperator{\sech}{sech}
\newtheorem{thm}{Theorem}[section]

\newtheorem{prop}[thm]{Proposition}

\allowdisplaybreaks[4]

\def\i{{\rm i}}

\def\l{\lambda}

\journal{}

\begin{document}
	
	\begin{frontmatter}
		\title{Soliton, breathers, positons and rogue waves for the vector complex modified Korteweg–de Vries equation}
		
		\author[1]{Yihang Liu}
		\author[2]{Yongshuai Zhang}
		\author[1]{\corref{cor1} Maohua Li}
		\address[1]{School of Mathematics and Statistics, Ningbo University,
			Ningbo, 315211, P.\ R.\ China}
		\address[2]{Department of Mathematics, Shaoxing University, Shaoxing, 312000, Zhejiang Province, P.\ R.\ China}
		\cortext[cor1]{Corresponding author: limaohua@nbu.edu.cn}
		
		\begin{abstract}
			This paper constructs the $N$-fold Darboux transformation (DT) for the vector complex modified Korteweg–de Vries (vcmKdV) equation and presents its determinant representation. Utilizing the DT and multi-fold eigenvalue degeneracy, we derive globally bounded solutions for the vcmKdV equation, including $N$-bright-bright-bright solitons, $N$-dark-bright-bright solitons, $N$-breathers, $N$-positon solutions, and $N$th-order rogue wave solutions." All these solutions are globally bounded. Graphical representations of bright-bright-bright and dark-bright-bright soliton solutions are provided, illustrating phenomena where periodic oscillatory waves coexist or interact with solitons. The collision scenarios of the two-bright-bright-bright solution have been investigated by using the asymptotic analysis. The bounded Akhmediev breather, the bounded breather with dark-bright soliton and breather-breather mixed waves are graphically shown. We give the graphs of the positon solution, the rogue wave and the rogue wave mixes with dark-bright solitons and breathers.
		\end{abstract}
	\end{frontmatter}
	
	\section{Introduction}
	\numberwithin{equation}{section}
	The Korteweg–de Vries (KdV) equation is widely recognized as a foundational and central model in nonlinear science and mathematical physics, and has the form
	\begin{equation}\label{a1}
		u_t+6uu_x+u_{xxx}=0,
	\end{equation}
	where $u=u(x,t)$ represents the wave field, and subscripts denote partial derivatives. This equation is completely integrable, up to now, there are many studies on this equation and the other nonlinear wave equations relate to KdV equation\cite{Pan1998,Clifford1967,Dejak2006,Ablositz1981,Stephen2011}. Among them, there emerges the modified Korteweg-de Vries (mKdV) equation
	\begin{equation}
		u_t+u_{xxx}+6u^{2}u_x=0\label{a2}.
	\end{equation}
	
	The mKdV equation has been extensively analyzed and extended to various contexts. Such as the  present multicomponent dusty plasma with ions, Cairns distributed electrons and immobile dusts\cite{Das2024}, the hyperbolic surfaces\cite{Schief1995}, etc\cite{GESZTESY1991,Salas2010,Nagatani1998,He2005,Liu2024,Zhang2020,Wu2017,Vladimir2023,Liu2021,Roy2022,Chen2023,Muhammad2024,Gao2021}. In recent years, significant attention has been directed toward studying this equation in the complex domain, and the complex modified Korteweg–de Vries (cmKdV) equation can be written as
	\begin{equation}
		u_t+u_{xxx}+6\left| u \right|^2u_x=0,\label{a3}
	\end{equation}
	where $u=u(x,t)$ is a complex function. The cmKdV equation has considerable significance in physic, and several solutions like rogue wave solution\cite{Zhaqilao2013}, Soliton molecules, rational positons\cite{Huang2021} and so on had been discussed\cite{Ismail2008,Uddin2009,Yuan2023,Tao2021,Wangwazwaz2022,Bai2024,Zhao2024,Song2024,Xu2025,Rao2024}. Furthermore, multi-component equations have attracted widespread interest, and there is vector modified Korteweg-de Vries (vmKdV) equation\cite{Liu2016}
	\begin{equation}
		q_t+q_{xxx}-3\delta \left( q_xq^Tq+qq^Tq_x \right) =0,\quad \delta =\pm 1\label{a4},
	\end{equation}  
	where $q$ is a $n\times 1\left( n\ge 2 \right)$ matrix real-valued function. The vmKdV has a $(n+1) \times (n+1)$~Lax pair. The vmKdV equation has more applications and can be applied to many more fields. Many solutions had been gotten by Riemann–Hilbert method, Fokas unified transform method and etc\cite{Wang2020,Fenchenko2018}.
	
	In this paper, we concentrate on the vector complex modified KdV (vcmKdV) equation
	\begin{equation}
		q_t+q_{xxx}+3\left( q_xq^{\dagger}q+qq^{\dagger}q_x \right) =0,\label{a5}
	\end{equation}
	where $q =(u,v,w)^{\mathsf{T}}$ is a $3 \times 1$ matrix complex-valued function, and the $\dagger$ denotes the Conjugate transpose.To the best of our knowledge, the three-component form of the vcmKdV equation has not been thoroughly investigated. In this paper we concentrate on the $N$-bright-bright-bright soliton, $N$-dark-bright-bright soliton, $N$-breathers, $N$-positon solutions and $N$-rogue wave solutions for Eq.\eqref{a5} by Darboux transformation and limit technique.
	
	The concept of the positon solution was introduced by Matveev in 1992 through the study of the Korteweg–de Vries (KdV) equation by applying a degeneration technique to the Wronskian formula solutions\cite{Matveev1992}. In subsequent studies, the investigation of positons has garnered considerable attention and has been extended to a wide range of equations\cite{Beutler,Stahlhofen1992,Shan2024,Hu2021,Rahman2025}. The positon can be characterized as ‘positons are long-range analogues of solitons and are slowly decreasing,  oscillating solutions’\cite{Matveev2002}. Positons are completely transparent regarding the collisional mechanism between soliton and positon, and positons keep unchanged after colliding with others. Furthermore, rogue waves have attracted a notable surge of research interest over the past decade. They are commonly described as waves that 'appear from nowhere and disappear without a trace' \cite{Akhmediev2009,Akhmediev2023},  and the amplitude is two or three times higher than the amplitude of background wave. The rogue waves are important phenomena and have lots of applications to many other fields,  like superthermal electrons\cite{Kaur2022}, encompassing oceanography\cite{Da2020}, nonlinear optics\cite{Qi2025}, Bose-Einstein condensation\cite{Kengne2023}, machine learning\cite{Haefner2023}, superfluid helium\cite{Efimov2010}, etc. Moreover, a variety of mathematical models have been developed to study rogue waves, like the complex mKdV equation\cite{Zhen2023}, sine-Gordon equation\cite{Li2020}, variable-coefficients high-order nonlinear Schr\"odinger equation\cite{Ying2021}, Ito's system\cite{Wang2024}.
	
	This article is structured into the following sections. In Section \ref{1}, we presented the Lax pair of the Eq.\eqref{a5} and derived both the first-order Darboux transformation and the determinant representation of the $N$-th order Darboux transformation. In Section \ref{2}, we obtain the bounded $N$-bright-bright-bright soliton solution and $N$-dark-bright-bright soliton solution with the vanishing background and the nonvanishing background, we find that there is a periodic oscillatory waves coexist or interact with the soliton solutions. Also we perform the asymptotic analysis for the $N$-bright-bright-bright soliton solution. In Section \ref{3}, we obtain the bounded $N$-breather solution with completely nonvanishing background, we obtain the bounded Akhmediev breather and collapsed Akhmediev breather by choosing special parameters, and we find that with different parameters there are three-component solution including a bounded Akhmediev breather with two dark-bright soliton solutions, a bounded Akhmediev brether with two breathers and a bounded Akhmediev brether with two soliton-breather-breather solutions. The graphs of the breather and mix breather-soliton solutions are presented. In Section \ref{4}, we construct the the determinant representation of $N$th-order positon soluton and rogue wave solution by using Darboux transformation with multi-fold degeneration of the eigenvalues. The graphs of the bounded positon solutions and rugue wave solutions are presented. And show the bounded rogue waves mix with soliton solutions and breathers. In Section \ref{5}, we summarize the conclusions.

	\section{Darboux transformation of the vcmKdV equation} \label{1}
	\numberwithin{equation}{section}
	\subsection{One-fold Darboux transformation}
	
	The Lax pair of the vector complex modified KdV equation \eqref{a5} can be given as
	
	\begin{align}
		&\Psi_x=U\Psi,U=-\i\zeta J+Q,\label{b1}\\
		&\Psi_t=V\Psi,V=-4\i\zeta^3 J+4\zeta^2 Q-2\i\zeta V_1+V_2,\label{b2}
	\end{align}
	with
	\begin{equation}\label{b3}
	    V_1=\left( Q^2+Q_x \right) J,\quad
	    V_2=-Q_{xx}+2Q^3+Q_xQ-QQ_x,
	\end{equation}
	where 
	\begin{align*}
		Q=\left( \begin{matrix}
			 0&		u&		v&		w\\
			-u^{*}&		0&		0&		0\\
			-v^{*}&		0&		0&		0\\
			-w^{*}&		0&		0&		0\\
		\end{matrix} \right), &&
		J=\left( \begin{matrix}
			1&		0&		0&		0\\
			0&		-1&		0&		0\\
			0&		0&		-1&		0\\
			0&		0&		0&		-1\\
		\end{matrix} \right) .
	\end{align*}
	 $\Psi = (\phi(x,t), \varphi(x,t), \psi(x,t), \chi(x,t))^T$ is the eigenfunction and $\zeta$ is the eigenvalue. We can find that the matrix $U$ and $V$ have the following property
	 \begin{align*}
	 	&U^{\dag}\left( x,t,\zeta ^* \right) =-U\left( x,t,\zeta \right),\\
	 	&V^{\dag}\left( x,t,\zeta ^* \right) =-V\left( x,t,\zeta \right).
	 \end{align*}
	So that the $\Psi ^{\dag}\left( x,t,\zeta ^* \right) $ satisfies the adjoint equations
	\begin{align*}
		&\Psi _{\text{x}}^{-1}=-\Psi ^{-1}U,\\
		&\Psi _{\text{t}}^{-1}=-\Psi ^{-1}V.
	\end{align*}
    The one-fold Darboux transformation can then be derived using the loop group method\cite{Terng2000,Wang2022}.
    \begin{prop}
    	The one-fold DT:
    	\begin{align}\label{b4}
    		\Psi[1] = T[1] \Psi, \quad
    		T[1] = I - \frac{\zeta_1 - \zeta_1^*}{\zeta - \zeta_1^*}\frac{\Psi_1 \Psi_1^{\dagger}}{\Psi_1^{\dagger} \Psi_1},
    	\end{align}
    	where the $\Psi_1 = (\phi_1(x,t), \varphi_1(x,t), \psi_1(x,t), \chi_1(x,t))^T$ is a solution of the \eqref{b1} and \eqref{b2} at $\zeta=\zeta_1$, and I is a 4$\times$4 identidy matrix. After the transformation, we can get
    	\begin{align}
    		u[1] &= u - 2\i (\zeta_1 - \zeta_1^*) \frac{\phi_1 \varphi_1^*}{\Psi_1^\dagger \Psi_1}, \label{b5}\\
    		v[1] &= u - 2\i (\zeta_1 - \zeta_1^*) \frac{\phi_1 \psi_1^*}{\Psi_1^\dagger \Psi_1},\label{b6} \\
    		w[1] &= u - 2\i (\zeta_1 - \zeta_1^*) \frac{\phi_1 \chi_1^*}{\Psi_1^\dagger \Psi_1}.\label{b7}
    	\end{align}
    \end{prop}
    Then, we set the $\Psi_i = \big( \phi_i(x,t), \varphi_i(x,t), \psi_i(x,t), \chi_i(x,t) \big)^T \ (i=1,\ldots,N)$ as the $N$ solutions of the linear matrix eigenvalue problem \eqref{b1} and \eqref{b2} when $\zeta=\zeta_i$. Based on this, the expression for $N$-fold DT can be derived. 
    \begin{prop}
    	The N-fold DT can be expressed as
    	\begin{align}\label{b8}
    		T_{N} = I - X M^{-1} S^{-1} X^\dagger,
    	\end{align}
    	where
        \begin{align*}
	       X &= (\Psi_1, \Psi_2, \ldots, \Psi_N), \\
	       S &= \mathrm{diag}\left\{ \zeta - \zeta_1^*, \zeta - \zeta_2^*, \ldots, \zeta - \zeta_N^* \right\},
        \end{align*}
    	\begin{equation}\label{b9}
    			M = \begin{pmatrix}
    				\dfrac{1}{\zeta_1 - \zeta_1^*} \Psi_1^\dagger \Psi_1 &
    				\dfrac{1}{\zeta_2 - \zeta_1^*} \Psi_1^\dagger \Psi_2 &
    				\cdots &
    				\dfrac{1}{\zeta_N - \zeta_1^*} \Psi_1^\dagger \Psi_N \\
    				
    				\dfrac{1}{\zeta_1 - \zeta_2^*} \Psi_2^\dagger \Psi_1 &
    				\dfrac{1}{\zeta_2 - \zeta_2^*} \Psi_2^\dagger \Psi_2 &
    				\cdots &
    				\dfrac{1}{\zeta_N - \zeta_2^*} \Psi_2^\dagger \Psi_N \\
    				
    				\vdots &
    				\vdots &
    				\ddots &
    				\vdots \\
    				
    				\dfrac{1}{\zeta_1 - \zeta_N^*} \Psi_N^\dagger \Psi_1 &
    				\dfrac{1}{\zeta_2 - \zeta_N^*} \Psi_N^\dagger \Psi_2 &
    				\cdots &
    				\dfrac{1}{\zeta_N - \zeta_N^*} \Psi_N^\dagger \Psi_N
    			    \end{pmatrix}.
    	\end{equation}\\
    	Concomitantly, the potentials after the transform can be formulated as
    	\begin{align}
    		u_N &= u + 2\i \frac{\det(M_1)}{\det(M)},\label{b10} \\
    		v_N &= v + 2\i \frac{\det(M_2)}{\det(M)},\label{b11} \\
    		w_N &= w + 2\i \frac{\det(M_3)}{\det(M)},\label{b12}
    	\end{align}\\
    	where
    	\begin{equation*}
    		M_1 = \begin{pmatrix}
    			\dfrac{1}{\zeta_1 - \zeta_1^*} \Psi_1^\dagger \Psi_1 &
    			\dfrac{1}{\zeta_2 - \zeta_1^*} \Psi_1^\dagger \Psi_2 &
    			\cdots &
    			\dfrac{1}{\zeta_N - \zeta_1^*} \Psi_1^\dagger \Psi_N &
    			\varphi_1^* \\
    			
    			\dfrac{1}{\zeta_1 - \zeta_2^*} \Psi_2^\dagger \Psi_1 &
    			\dfrac{1}{\zeta_2 - \zeta_2^*} \Psi_2^\dagger \Psi_2 &
    			\cdots &
    			\dfrac{1}{\zeta_N - \zeta_2^*} \Psi_2^\dagger \Psi_N &
    			\varphi_2^* \\
    			
    			\vdots &
    			\vdots &
    			\ddots &
    			\vdots &
    			\vdots \\
    			
    			\dfrac{1}{\zeta_1 - \zeta_N^*} \Psi_N^\dagger \Psi_1 &
    			\dfrac{1}{\zeta_2 - \zeta_N^*} \Psi_N^\dagger \Psi_2 &
    			\cdots &
    			\dfrac{1}{\zeta_N - \zeta_N^*} \Psi_N^\dagger \Psi_N &
    			\varphi_N^* \\
    			
    			\phi_1 &
    			\phi_2 &
    			\cdots &
    			\phi_N &
    			0
    		\end{pmatrix},
    	\end{equation*}
    	\begin{equation*}
    		M_2 = \begin{pmatrix}
    			\dfrac{1}{\zeta_1 - \zeta_1^*} \Psi_1^\dagger \Psi_1 &
    			\dfrac{1}{\zeta_2 - \zeta_1^*} \Psi_1^\dagger \Psi_2 &
    			\cdots &
    			\dfrac{1}{\zeta_N - \zeta_1^*} \Psi_1^\dagger \Psi_N &
    			\psi_1^* \\
    			
    			\dfrac{1}{\zeta_1 - \zeta_2^*} \Psi_2^\dagger \Psi_1 &
    			\dfrac{1}{\zeta_2 - \zeta_2^*} \Psi_2^\dagger \Psi_2 &
    			\cdots &
    			\dfrac{1}{\zeta_N - \zeta_2^*} \Psi_2^\dagger \Psi_N &
    			\psi_2^* \\
    			
    			\vdots &
    			\vdots &
    			\ddots &
    			\vdots &
    			\vdots \\
    			
    			\dfrac{1}{\zeta_1 - \zeta_N^*} \Psi_N^\dagger \Psi_1 &
    			\dfrac{1}{\zeta_2 - \zeta_N^*} \Psi_N^\dagger \Psi_2 &
    			\cdots &
    			\dfrac{1}{\zeta_N - \zeta_N^*} \Psi_N^\dagger \Psi_N &
    			\psi_N^* \\
    			
    			\phi_1 &
    			\phi_2 &
    			\cdots &
    			\phi_N &
    			0
    		\end{pmatrix},
    	\end{equation*}
    	\begin{equation*}
    		M_3 = \begin{pmatrix}
    			\dfrac{1}{\zeta_1 - \zeta_1^*} \Psi_1^\dagger \Psi_1 &
    			\dfrac{1}{\zeta_2 - \zeta_1^*} \Psi_1^\dagger \Psi_2 &
    			\cdots &
    			\dfrac{1}{\zeta_N - \zeta_1^*} \Psi_1^\dagger \Psi_N &
    			\chi_1^* \\
    			
    			\dfrac{1}{\zeta_1 - \zeta_2^*} \Psi_2^\dagger \Psi_1 &
    			\dfrac{1}{\zeta_2 - \zeta_2^*} \Psi_2^\dagger \Psi_2 &
    			\cdots &
    			\dfrac{1}{\zeta_N - \zeta_2^*} \Psi_2^\dagger \Psi_N &
    			\chi_2^* \\
    			
    			\vdots &
    			\vdots &
    			\ddots &
    			\vdots &
    			\vdots \\
    			
    			\dfrac{1}{\zeta_1 - \zeta_N^*} \Psi_N^\dagger \Psi_1 &
    			\dfrac{1}{\zeta_2 - \zeta_N^*} \Psi_N^\dagger \Psi_2 &
    			\cdots &
    			\dfrac{1}{\zeta_N - \zeta_N^*} \Psi_N^\dagger \Psi_N &
    			\chi_N^* \\
    			
    			\phi_1 &
    			\phi_2 &
    			\cdots &
    			\phi_N &
    			0
    		\end{pmatrix}.
    	\end{equation*}
    \end{prop}
    \begin{proof}
    	
    	By using the One-fold DT \eqref{b4}, we can get
    	\begin{flalign}\label{b13}
    		T_N &=T[N] T[N-1] \cdots T[1] \notag&\\
    		&= \left( I - \frac{\zeta_N - \zeta_N^*}{\zeta - \zeta_N^*} \frac{\Psi_N \Psi_N^\dagger}{\Psi_N^\dagger \Psi_N} \right) \left( I - \frac{\zeta_{N-1} - \zeta_{N-1}^*}{\zeta - \zeta_{N-1}^*} \frac{\Psi_{N-1} \Psi_{N-1}^\dagger}{\Psi_{N-1}^\dagger \Psi_{N-1}} \right) \cdots 
    		\left( I - \frac{\zeta_1 - \zeta_1^*}{\zeta - \zeta_1^*} \frac{\Psi_1 \Psi_1^\dagger}{\Psi_1^\dagger \Psi_1} \right) \notag &\\ 
    		&= I - \sum_{i=1}^N \frac{D_i}{\zeta - \zeta_i^*}, &
    	\end{flalign}
    	where $D_i$ is a uncertain matrix.
    	
    	We can use the Residue Theorem to obtain the expression $D_i$.
    	\begin{align*}
    		D_i = T_N \Big|_{Res\,\zeta = \zeta_i^*} = T[N] \cdots T[i+1] \cdot \frac{\zeta_i - \zeta_i^*}{\Psi_i^\dagger \Psi_i} \Psi_i \Psi_i^\dagger \cdot T[i-1] \cdots T[1] \Big|_{\zeta = \zeta_i^*}.
    	\end{align*}
    	Set $\left| z_i \right\rangle = \left. T[N] \cdots T[i+1] \cdot \frac{\zeta_i - \zeta_i^*}{\Psi_i^\dagger \Psi_i} \Psi_i \right|_{\zeta = \zeta_i^*}$, $\left\langle y_i \right|^{T} = \left. \Psi_i^\dagger \cdot T[i-1] \cdots T[1] \right|_{\zeta = \zeta_i^*}$. The $\left| z_i \right\rangle$ and $\left\langle y_i \right|$ are column vectors. Then we can get
    	\begin{align*}
    		T_N = I - \sum_{i=1}^N \frac{\left| z_i \right\rangle \left\langle y_i \right|^{T}}{\zeta - \zeta_i^*},&\quad T_N^{-1} = I - \sum_{i=1}^N \frac{\left\langle y_i \right|^* \left| z_i \right\rangle^\dagger}{\zeta - \zeta_i}.
    	\end{align*}
    	Through the Residue Theorem we can get
    	\begin{align}\label{b14}
    		&Res_{\zeta = \zeta_j} \left[ I - \sum_{i=1}^N \frac{\left| z_i \right\rangle \left\langle y_i \right|^{T}}{\zeta - \zeta_i^*} \right] \left[ I - \sum_{i=1}^N \frac{\left\langle y_i \right|^* \left| z_i \right\rangle^\dagger}{\zeta - \zeta_i} \right] = -\left[ I - \sum_{i=1}^N \frac{\left| z_i \right\rangle \left\langle y_i \right|^{T}}{\zeta_j - \zeta_i^*} \right] \left\langle y_j \right|^* \left| z_j \right\rangle^\dagger = 0.
    	\end{align}
    	Then we can get the correlation of the $\left| z_i \right\rangle$ and $\left\langle y_i \right|$.
    	\begin{align*}
    		\left\langle y_j \right|^* = \sum_{i=1}^N \frac{\left| z_i \right\rangle \left\langle y_i \right|^{T} \left\langle y_j \right|^*}{\zeta_j - \zeta_i^*},
    	\end{align*}
    	which implies
    	\begin{align*}
    		\left( \left\langle y_1 \right|, \left\langle y_2 \right|, \dots, \left\langle y_N \right| \right)^* 
    		&= \left( \left| z_1 \right\rangle, \left| z_2 \right\rangle, \dots, \left| z_N \right\rangle \right) 
    		\begin{pmatrix}
    			\dfrac{1}{\zeta_1 - \zeta_1^*} \left\langle y_1 \right|^{T} \left\langle y_1 \right|^* 
    			& \dfrac{1}{\zeta_2 - \zeta_1^*} \left\langle y_1 \right|^{T} \left\langle y_2 \right|^* 
    			& \cdots 
    			& \dfrac{1}{\zeta_N - \zeta_1^*} \left\langle y_1 \right|^{T} \left\langle y_N \right|^* \\[2ex]
    			\dfrac{1}{\zeta_1 - \zeta_2^*} \left\langle y_2 \right|^{T} \left\langle y_1 \right|^* 
    			& \dfrac{1}{\zeta_2 - \zeta_2^*} \left\langle y_2\right|^{T} \left\langle y_2 \right|^* 
    			& \cdots 
    			& \dfrac{1}{\zeta_N - \zeta_2^*} \left\langle y_2\right|^{T} \left\langle y_N \right|^* \\[2ex]
    			\vdots 
    			& \vdots 
    			& \ddots 
    			& \vdots \\[2ex]
    			\dfrac{1}{\zeta_1 - \zeta_N^*} \left\langle y_N\right|^{T} \left\langle y_1 \right|^* 
    			& \dfrac{1}{\zeta_2 - \zeta_N^*} \left\langle y_N\right|^{T} \left\langle y_2 \right|^* 
    			& \cdots 
    			& \dfrac{1}{\zeta_N - \zeta_N^*} \left\langle y_N\right|^{T} \left\langle y_N \right|^*
    		\end{pmatrix}\\
    		&\triangleq \left( \left| z_1 \right\rangle, \left| z_2 \right\rangle, \dots, \left| z_N \right\rangle \right)M
    	\end{align*}
    	And by \eqref{b14} we have
    	\begin{equation*}
    		T_N \left< y_j \right|^T \left| z_j \right> ^{\dagger} \bigg|_{\zeta = \zeta_j} = 0.
    	\end{equation*}
    	So we can set $\left< y_j \right|^T = \Psi_j = \Psi(x, t, \zeta_j)$ without loss of generality. Then we can obtain
    	\begin{align*}
    		\left( |z_1\rangle, |z_2\rangle, \ldots, |z_N\rangle \right) = \left( \Psi_1, \Psi_2, \ldots, \Psi_N \right) M^{-1}.
    	\end{align*}
    	Then we can get the express of $N$-fold DT as \eqref{b8}. About the potentials after the $N$-times transforms, we can consider the equation
    	\begin{align}\label{b15}
    		(T_N)_x + T_N U = U[N] T_N.
    	\end{align}
    	And the different between the matrix $U$ and the $U[N]$ are only the change of potentials. By equation \eqref{b15} we can get
    	\begin{align*}
    		-\sum_{i=1}^{N} \frac{\mathrm{D}_{ix}}{\zeta - \zeta_i^*} 
    		+ \left( I - \sum_{i=1}^{N} \frac{\mathrm{D}_i}{\zeta - \zeta_i^*} \right) (-\i\zeta J - Q) 
    		= (-\i\zeta J - Q[N]) \left( I - \sum_{i=1}^{N} \frac{\mathrm{D}_i}{\zeta - \zeta_i^*} \right).
    	\end{align*}
    	Taking the limit as $\zeta \rightarrow \infty$ and considering the coefficient of $\zeta ^0$, we have
    	\begin{align*}
    		Q[N] = Q + \i \sum_{i=1}^N [D_i, J].
    	\end{align*}
    	This yields
    	\begin{align*}
    		u[N] = u - 2\i \left( \sum_{j=1}^N D_j \right)_{12},\quad
    		v[N] = v - 2\i \left( \sum_{j=1}^N D_j \right)_{13},\quad
    		w[N] = w - 2\i \left( \sum_{j=1}^N D_j \right)_{14}.
    	\end{align*}
    	By calculating we can get 
    	\begin{align*}
    		\left( \sum_{i=1}^N D_i \right)_{\!12} 
    		= \dfrac{
    			\displaystyle \sum_{j=1}^N \sum_{i=1}^N \phi_i M_{ij}^* \varphi_j^*
    		}{
    			\det(M)
    		}
    		= -\dfrac{
    			\det(M_1)
    		}{
    			\det(M)
    		}.
    	\end{align*}
    	By similarly calculations, we can obtain the results for equations \eqref{b10}-\eqref{b12}.
    \end{proof}
    
    \section{Soliton solution}\label{2}
    \numberwithin{equation}{section}
    \subsection{Bright-bright-bright soliton solution}
    In this section, we discuss the $N$-fold DT applied to zero seed, with the $u= 0, v=0, w=0$, the solution of the Lax pairs at $\zeta=\zeta_k$ is given by
    \begin{equation}\label{c1}
    	\Psi_k=\left(
    	\begin{array}{c}
    		e^{ -\i \zeta_k (4t \zeta_k^2 + x) } \\
    		c_1 e^{ \i \zeta_k (4t \zeta_k^2 + x) } \\
    		c_2 {e^{ \i \zeta_k (4t \zeta_k^2 + x) }} \\
    		c_3 e^{\i \zeta_k (4t \zeta_k^2 + x) }
    	\end{array}
    	\right).
    \end{equation}
    The $c_1$, $c_2$ and the $c_3$ are the arbitrary constants, through the application of equations \eqref{b10}-\eqref{b12}. We have
    \begin{flalign}
    	u_N &= 2\i \frac{\det(M_1)}{\det(M)},\label{c2} \\
    	v_N &= 2\i \frac{\det(M_2)}{\det(M)}, \label{c3}\\
    	w_N &= 2\i \frac{\det(M_3)}{\det(M)},\label{c4}
    \end{flalign}
    where
    \begin{align*}
    	M = \bigl( M_{ij} \bigr)_{1 \le i,j \le N},\quad
    	M_1 = \begin{pmatrix}
    		M & Y_1 \\
    		X_1 & 0 
    	\end{pmatrix},\quad 
    	M_2 = \begin{pmatrix}
    		M & Y_2 \\
    		X_1 & 0 
    	\end{pmatrix},\quad 
    	M_3 = \begin{pmatrix}
    		M & Y_3 \\
    		X_1 & 0 
    	\end{pmatrix},
    \end{align*}
    and
    \begin{flalign*}
    	&M_{ij}=\frac{
    		\bigl( c_1^2 + c_2^2 + c_3^2 \bigr)
    		\left[
    		e^{4\i\left( (\zeta_j^2 + \zeta_j\zeta_i^* + (\zeta_i^*)^2 )t + \frac{x}{4} \right)(\zeta_j - \zeta_i^*)} 
    		+ e^{-4\i\left( (\zeta_j^2 + \zeta_j\zeta_i^* + (\zeta_i^*)^2 )t + \frac{x}{4} \right)(\zeta_j - \zeta_i^*)}
    		\right]
    	}{\zeta_j - \zeta_i^*},&\\
    	&X_1 = \begin{pmatrix}
    		e^{-\i\zeta_1 (4t\zeta_1 + x)} 
    		& e^{-\i\zeta_2 (4t\zeta_2 + x)} 
    		& \cdots 
    		& e^{-\i\zeta_N (4t\zeta_N + x)}
    	\end{pmatrix},&\\
    	&Y_1 = \begin{pmatrix}
    		c_1 e^{-\i\zeta_1^* (4t\zeta_1^* + x)} 
    		&c_1 e^{-\i\zeta_2^* (4t\zeta_2^* + x)}
    		&\cdots 
    		&c_1 e^{-\i\zeta_N^* (4t\zeta_N^* + x)}
    	\end{pmatrix}^{T},&\\
    	&Y_2 = \begin{pmatrix}
    		c_2 e^{-\i\zeta_1^* (4t\zeta_1^* + x)} 
    		&c_2 e^{-\i\zeta_2^* (4t\zeta_2^* + x)} 
    		&\cdots 
    		&c_2 e^{-\i\zeta_N^* (4t\zeta_N^* + x)}
    	\end{pmatrix}^{T},&\\
    	&Y_3 = \begin{pmatrix}
    		c_3 e^{-\i\zeta_1^* (4t\zeta_1^* + x)} 
    		&c_3 e^{-\i\zeta_2^* (4t\zeta_2^* + x)} 
    		&\cdots 
    		&c_3 e^{-\i\zeta_N^* (4t\zeta_N^* + x)}
    	\end{pmatrix}^{T}.
    \end{flalign*}
    When $N$=1, the solutions of vcmKdV generated by one-fold DT can be expressed as
    \begin{align}
    	&u[1] = \frac{-\i c_1 (\zeta_1 - \zeta_1^*)}{\sqrt{c_1^2 + c_2^2 + c_3^2}} 
    	e^{-4\i (\zeta_1 + \zeta_1^*)\left[ (\zeta_1^2 - \zeta_1 \zeta_1^* + (\zeta_1^*)^2 )t + \frac{x}{4} \right]} 
    	\sech\left[ 4\i (\zeta_1 - \zeta_1^*)\left[ (\zeta_1^2 + \zeta_1 \zeta_1^* + (\zeta_1^*)^2 )t + \frac{x}{4} \right] 
    	+ \frac{1}{2} \ln(c_1^2 + c_2^2 + c_3^2) \right],\label{c5}\\
    	&v[1] = \frac{-\i c_2 (\zeta_1 - \zeta_1^*)}{\sqrt{c_1^2 + c_2^2 + c_3^2}} 
    	e^{-4\i (\zeta_1 + \zeta_1^*)\left[ (\zeta_1^2 - \zeta_1 \zeta_1^* + (\zeta_1^*)^2 )t + \frac{x}{4} \right]} 
    	\sech\left[ 4\i (\zeta_1 - \zeta_1^*)\left[ (\zeta_1^2 + \zeta_1 \zeta_1^* + (\zeta_1^*)^2 )t + \frac{x}{4} \right] 
    	+ \frac{1}{2} \ln(c_1^2 + c_2^2 + c_3^2) \right],\label{c6} \\
    	&w[1] = \frac{-\i c_3 (\zeta_1 - \zeta_1^*)}{\sqrt{c_1^2 + c_2^2 + c_3^2}} 
    	e^{-4\i (\zeta_1 + \zeta_1^*)\left[ (\zeta_1^2 - \zeta_1 \zeta_1^* + (\zeta_1^*)^2 )t + \frac{x}{4} \right]} 
    	\sech\left[ 4\i (\zeta_1 - \zeta_1^*)\left[ (\zeta_1^2 + \zeta_1 \zeta_1^* + (\zeta_1^*)^2 )t + \frac{x}{4} \right] 
    	+ \frac{1}{2} \ln(c_1^2 + c_2^2 + c_3^2) \right].\label{c7}
    \end{align}
    
    \begin{figure}[tbh]
    	\centering
    	\subfigure[]{\includegraphics[height=4.8cm,width=4.8cm]{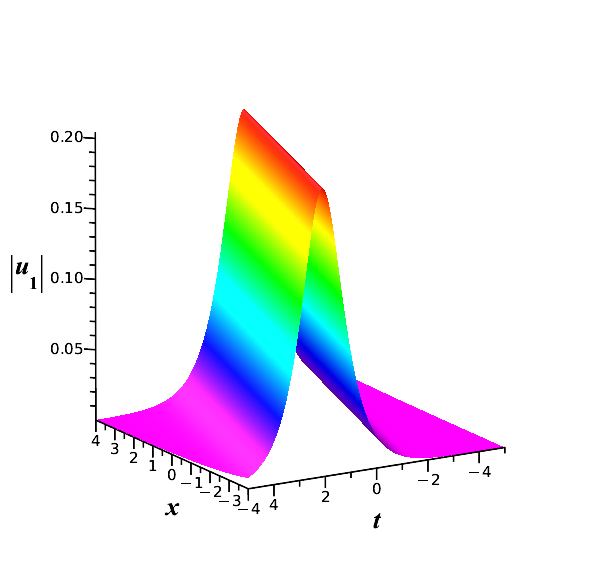}}
    	\subfigure[]{\includegraphics[height=4.8cm,width=4.8cm]{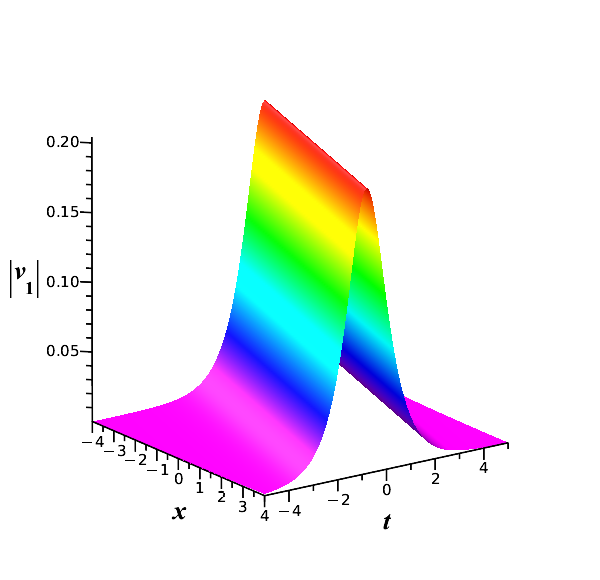}}
    	\subfigure[]{\includegraphics[height=4.8cm,width=4.8cm]{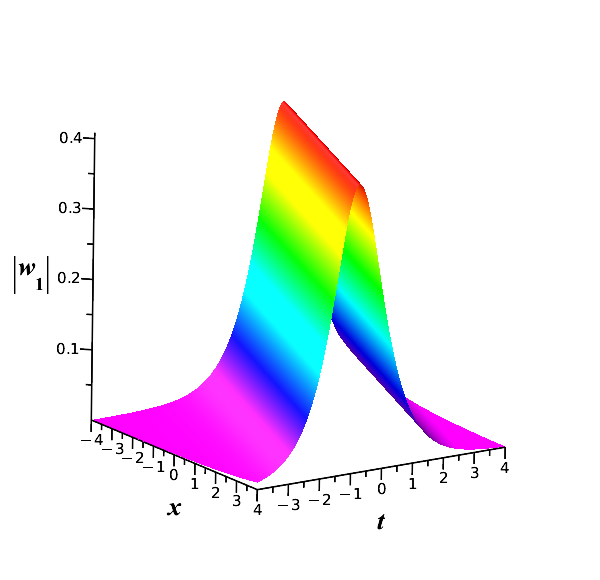}}
    	\caption{The bounded one-bright-bright-bright solution of vcmKdV with $\zeta_1 = \dfrac{1}{2} + \dfrac{1}{4}i$, $c_1=1, c_2=1, c_3=2$}
    	\label{s1-1}
    \end{figure}

    The solutions with $\zeta_1=a_1+\i b_1$ are linear solutions, see Fig. \ref{s1-1}, and we can see that there exists a proportional relationship among the components $u[1], v[1]$ and $w[1]$, such that $u\left[ 1 \right] :v\left[ 1 \right] :w\left[ 1 \right] =c_1:c_2:c_3$. And the trajectory of the solutions is given by 
    \begin{flalign}
       &\quad x = \left( 4b_1^2 - 12a_1^2 \right) t + \dfrac{\ln\left( c_1^2 + c_2^2 + c_3^2 \right)}{16b_1}, \label{c8}&
    \end{flalign}
    the maximum of the solutions are
    \begin{align*}
    	|u[1]|_{max} = \dfrac{2|c_1 b_1|}{\sqrt{c_1^2 + c_2^2 + c_3^2}},\quad |v[1]|_{max} = \dfrac{2|c_2 b_1|}{\sqrt{c_1^2 + c_2^2 + c_3^2}},\quad |w[1]|_{max} = \dfrac{2|c_3 b_1|}{\sqrt{c_1^2 + c_2^2 + c_3^2}}.
    \end{align*}
    When $N$=2, the expression of the solutions after two-fold DT are
    \begin{align}
    	&u[2] = 2\i \dfrac{c_1 \left[ 
    		M_{12} A_2^* A_1 
    		+ M_{21} A_1^* A_2 
    		- M_{11} A_2^* A_2^* 
    		- M_{22} A_1 A_1
    		\right]}{M_{11}M_{22} - M_{12}M_{21}},\label{c9}\\
    	&v[2] = 2\i \dfrac{c_2 \left[ 
    		M_{12} A_2^* A_1 
    		+ M_{21} A_1^* A_2 
    		- M_{11} A_2^* A_2^* 
    		- M_{22} A_1 A_1
    		\right]}{M_{11}M_{22} - M_{12}M_{21}},\label{c10} \\
    	&w[2] = 2\i \dfrac{c_3 \left[ 
    		M_{12}A_2^* A_1 
    		+ M_{21} A_1^* A_2
    		- M_{11} A_2^* A_2^* 
    		- M_{22} A_1 A_1
    		\right]}{M_{11}M_{22} - M_{12}M_{21}},\label{c11}
    \end{align}
    where
    \begin{align*}
    	A_1&=e^{-\i\zeta_1 \left(4t\zeta_1^2 + x\right)},\quad
    	A_2=e^{-\i\zeta_2 \left(4t\zeta_2^2 + x\right)},\\
    	M_{i,j} &= \dfrac{2\sqrt{c_1^2 + c_2^2 + c_3^2}}{\zeta_j - \zeta_i^*}
    	\cosh\left[ 4\i (\zeta_j - \zeta_i^*)\left( (\zeta_j^2 + \zeta_j \zeta_i^* + (\zeta_i^*)^2 )t + \dfrac{x}{4} \right) 
    	+ \dfrac{1}{2} \ln(c_1^2 + c_2^2 + c_3^2) \right].
    \end{align*} 
    
    \begin{figure}[tbh]
    	\centering
    	\subfigure[]{\includegraphics[height=6.8cm,width=6.8cm]{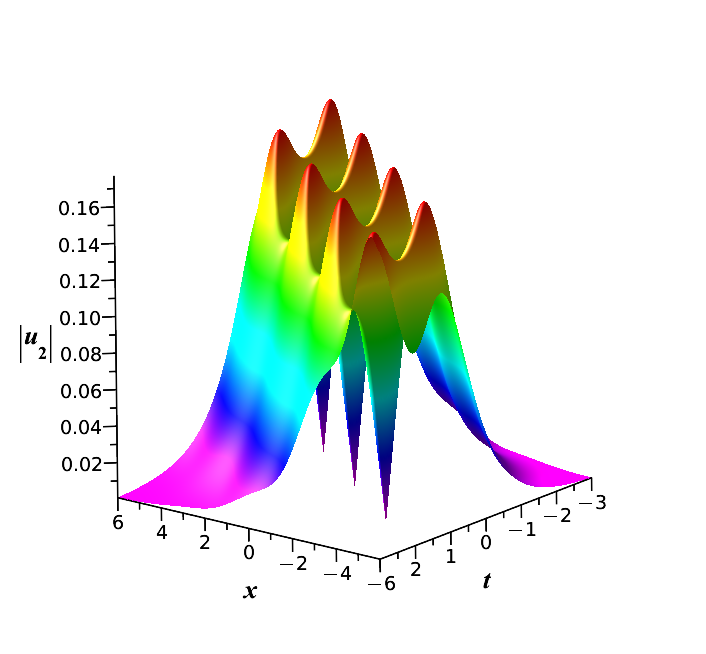}}
    	\subfigure[]{\includegraphics[height=6.8cm,width=6.8cm]{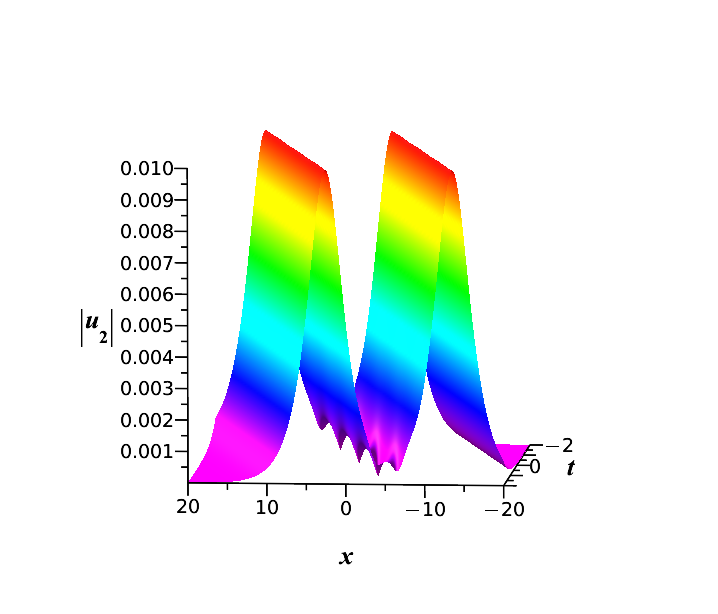}}
    	\caption{Solutions with the choice of $\zeta_1 = -\zeta_2$.Panel (a) is the bounded periodic solution with $\zeta_1 = \dfrac{1}{2} + \dfrac{1}{4}\i$, $\zeta_2 = -\dfrac{1}{2} - \dfrac{1}{4}\i$, $c_1 = 1$, $c_2 = 1$, $c_3 = 3$. Panel (b) is the bounded two bright solitons with a periodic solution with $\zeta_1 = \dfrac{1}{2} + \dfrac{1}{4}\i$, $\zeta_2 = -\dfrac{1}{2} - \dfrac{1}{4}\i$, $c_1 = 1$, $c_2 = 1$, $c_3 = 50$.}
    	\label{s1-21}
    \end{figure} 
   \begin{figure}[tbh]
   	\centering
   	\subfigure[]{\includegraphics[height=4.8cm,width=4.8cm]{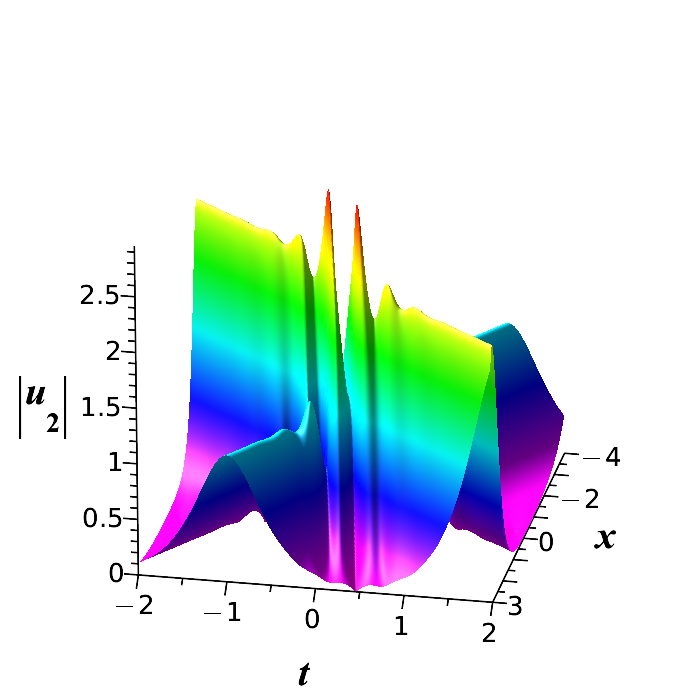}}
   	\subfigure[]{\includegraphics[height=4.8cm,width=4.8cm]{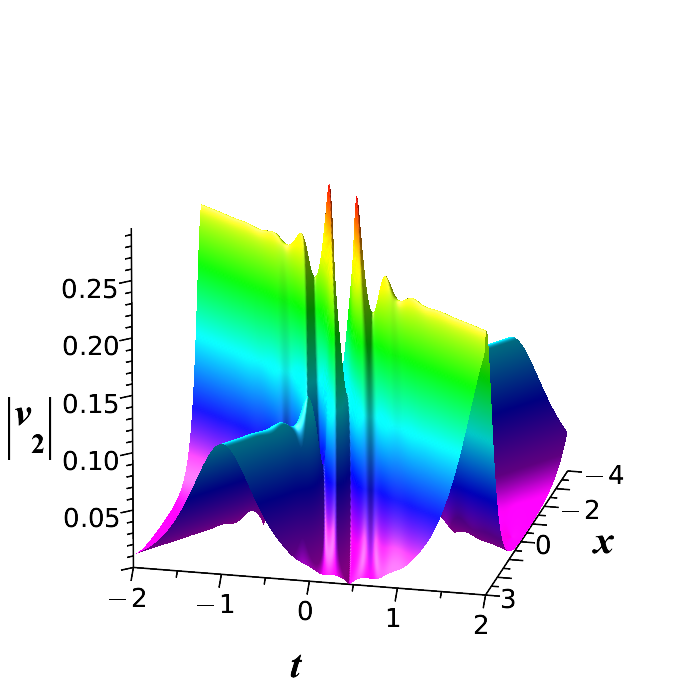}}
   	\subfigure[]{\includegraphics[height=4.8cm,width=4.8cm]{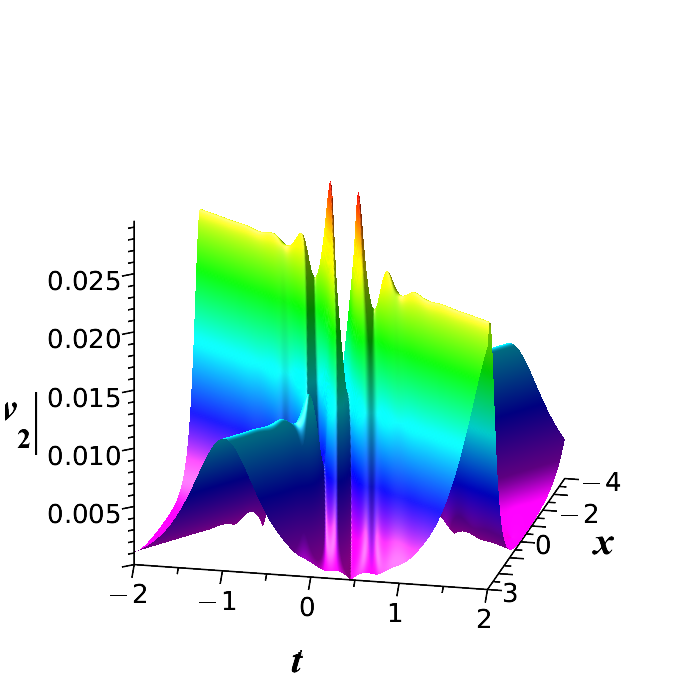}}
   	\subfigure[]{\includegraphics[height=4.8cm,width=4.8cm]{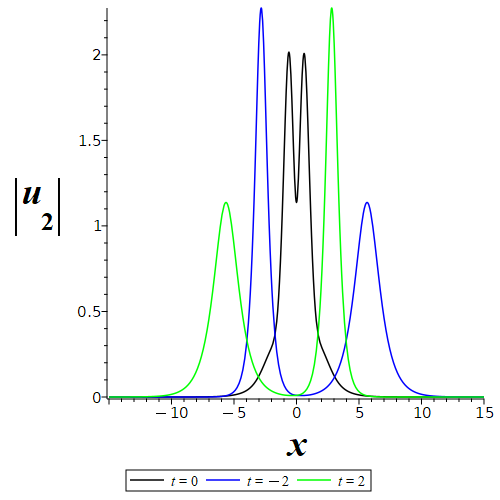}}
   	\subfigure[]{\includegraphics[height=4.8cm,width=4.8cm]{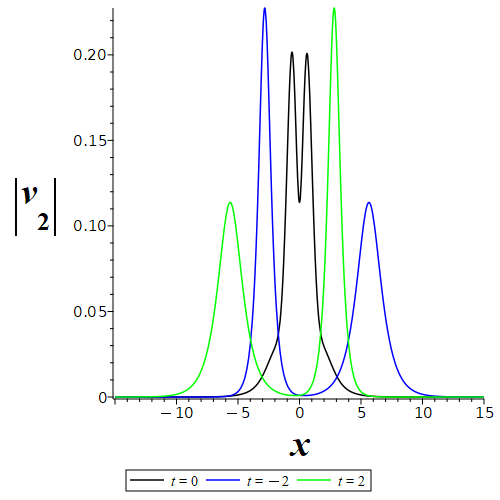}}
   	\subfigure[]{\includegraphics[height=4.8cm,width=4.8cm]{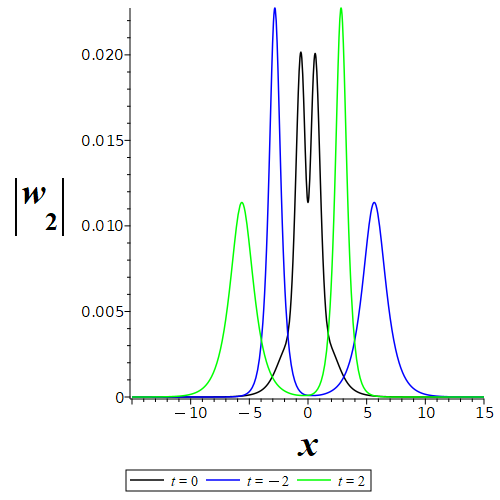}}
   \caption{Panel (a), (b), (c) are the 3D plots of the bounded two-bright-bright-bright solution of vcmKdV with $\zeta_1 = \dfrac{4}{7} + \dfrac{4}{7}i$, $\zeta_2 = -\dfrac{4}{7} - \dfrac{8}{7}i$, $c_1 = 1$, $c_2 = \dfrac{1}{10}$, $c_3 = \dfrac{1}{100}$. (d), (e), (f) are the profiles of the bounded two-bright-bright-bright solution at $t=0, t=-2, t=2$.}
   	\label{s1-2}
   	\end{figure} 
   The solutions generated by the two-fold DT maintain the same ratio among the components $u[1], v[1]$ and $w[1]$, and by analyzing the \eqref{c9}--\eqref{c11}, we observe that component $(\zeta_j - \zeta_i^*)$ in the functional structure possesses a critical property: selecting judicious values for specific parameters induces cancellation effects on certain terms. To systematically characterize this behavior, we establish the classification cases as follow
   
   (1) For $\zeta_1 = -\zeta_2^*$, then counter-diagonal elements in matrix $M$ is identically zero. As a result, we can get a bounded bright breather-like solution.
   
   (2) For $\zeta_1 = -\zeta_2$, the component $(\zeta_1 - \zeta_2^*)=(\zeta_1 + \zeta_1^*)$, and under this parameter choice the equations \eqref{c9}-\eqref{c11} will transform into a remarkably structured form, the solution $u[2]$ can be wrote as
   \begin{align}
   	  u[2] = \dfrac{
   	  	-2c_1 \bigl( \zeta_1 - \zeta_1^* \bigr) \bigl( \zeta_1 + \zeta_1^* \bigr) \bigl( k\zeta_1^* \sin z + \zeta_1 \sin z^* \bigr)
   	  }{
   	  	\biggl( \zeta_1^2 + \zeta_1^{*2} \biggr) k \sin z \sin z^*
   	  	+ 2\zeta_1 \zeta_1^* \biggl( k \cos z \cos z^* + \dfrac{k^2 + 1}{2} \biggr)
   	  },
   \end{align}
   where
   \begin{align*}
   	k=c_{1}^{2}+c_{2}^{2}+c_{3}^{2},\quad z=\alpha +\i\beta =8\zeta _{1}^{3}t+2\zeta _1x,\quad \zeta =a+\i b.
   \end{align*}
   We can get a bounded bright periodic solution, see Fig.\ref{s1-21} (a). And by choosing parameters $c_1, c_2, c_3$, we can obtain a bounded two-bright solitons coexit with a periodic solution, see Fig.\ref{s1-21} (b). The trajectory of the two bright solitons are
   \[
   \left\{
   \begin{aligned}
   	x &= -4\bigl(3ab - b^2\bigr) t + 2 \operatorname{arcsinh}\biggl( \frac{\bigl(a^2 + b^2\bigr) (k + 1)^2}{4a^2k} \biggr) \\
   	x &= -4\bigl(3ab - b^2\bigr) t - 2 \operatorname{arcsinh}\biggl( \frac{\bigl(a^2 + b^2\bigr) (k + 1)^2}{4a^2k} \biggr)
   \end{aligned}
   \right.
   \]
   and the maximum of the two bright solitons are
   \begin{align*}
   	\lvert u[2] \rvert_{\max} = \dfrac{2bc_1}{k+1} \sqrt{ \dfrac{ 
   			b^2 (k + 1)^2 + a^2 (k - 1)^2 
   		}{ 
   			k (a^2 + b^2) 
   	} }.
   \end{align*}
   We can find that except the solitons there is a periodic oscillatory waves, and we will present this periodic wave interact with the solitons in the following case.
   
   (3) For $\operatorname{Im}\zeta_1 = k \operatorname{Im}\zeta_2\ (k \neq 0)$, we will get the bounded two-bright-bright-bright solutions. Specifically, choosing $\zeta_1 = \tfrac{4}{7} + \tfrac{4}{7}\i,\ \zeta_2 = -\tfrac{4}{7} - \tfrac{8}{7}\i$
   , the solution can be plotted, see Fig. \ref{s1-2}. We can find that the solution have two bright solitons and the periodic oscillatory waves, the three-component waves interact solely during collision events, and throughout the interaction process, the periodic waves retain their inherent periodicity while achieving maximum amplitude.
   
   To investigate the collision scenarios of the solutions, we use asymptotic analysis. For this purpose, we choose $\zeta_1 = a_1 + a_1 \i$, $\zeta_2 = -a_1 - 2a_1 \i$ with $a_1 > 0$, and derive the following asymptotic expressions for the solutions.
   
   \noindent \textbf{Case 1}:When $x-4a_1^2t=c$ ($c$ is a constant)
   
   Before interacting ($t\rightarrow\infty$), we can get\\
   \begin{align}
   	  &u[2] \rightarrow \dfrac{\sqrt{5}(32\i - 56) a_1 c_1}{2\sqrt{13(c_1^2 + c_2^2 + c_3^2)}} \operatorname{sech}\left[ 4a_1 c + \dfrac{1}{2} \ln (c_1^2 + c_2^2 + c_3^2) +\dfrac{1}{2} \ln \left( \dfrac{13}{5} \right) \right] e^{-80\i a_1^3 t + 2\i a_1 c},\\
   	  &v[2] \rightarrow \dfrac{\sqrt{5}(32\i - 56) a_1 c_2}{2\sqrt{13(c_1^2 + c_2^2 + c_3^2)}} \operatorname{sech}\left[ 4a_1 c + \dfrac{1}{2} \ln (c_1^2 + c_2^2 + c_3^2) +\dfrac{1}{2} \ln \left( \dfrac{13}{5} \right) \right] e^{-80\i a_1^3 t + 2\i a_1 c},\\
   	  &w[2] \rightarrow \dfrac{\sqrt{5}(32\i - 56) a_1 c_3}{2\sqrt{13(c_1^2 + c_2^2 + c_3^2)}} \operatorname{sech}\left[ 4a_1 c + \dfrac{1}{2} \ln (c_1^2 + c_2^2 + c_3^2) +\dfrac{1}{2} \ln \left( \dfrac{13}{5} \right) \right] e^{-80\i a_1^3 t + 2\i a_1 c}.
   \end{align}
   
   After interacting ($t\rightarrow-\infty$), we can get
   
   \begin{align}
   	  &u[2] \rightarrow \dfrac{\sqrt{13}(32\i + 56)a_1c_1}{2\sqrt{5(c_1^2 + c_2^2 + c_3^2)}} \mathrm{sech}\left[ 4a_1c + \dfrac{1}{2} \ln (c_1^2 + c_2^2 + c_3^2) +\dfrac{1}{2} \ln \left( \dfrac{5}{13} \right) \right] e^{-80\i a_1^3 t + 2\i a_1c},\\
   	  &v[2] \rightarrow \dfrac{\sqrt{13}(32\mathrm{i} + 56)a_1c_2}{2\sqrt{5(c_1^2 + c_2^2 + c_3^2)}} \mathrm{sech}\left[ 4a_1c + \dfrac{1}{2} \ln (c_1^2 + c_2^2 + c_3^2) +\dfrac{1}{2} \ln \left( \dfrac{5}{13} \right) \right] e^{-80\mathrm{i}a_1^3 t + 2\mathrm{i}a_1c},\\
   	  &w[2] \rightarrow \dfrac{\sqrt{13}(32\mathrm{i} + 56)a_1c_3}{2\sqrt{5(c_1^2 + c_2^2 + c_3^2)}} \mathrm{sech}\left[ 4a_1c + \dfrac{1}{2} \ln (c_1^2 + c_2^2 + c_3^2) +\dfrac{1}{2} \ln \left( \dfrac{5}{13} \right) \right] e^{-80\mathrm{i}a_1^3 t + 2\mathrm{i}a_1c}.\\
   \end{align}
    \noindent \textbf{Case 2}:When $x+8a_1^2t=c$
    
    Before interacting($t\rightarrow\infty$), we can get
    
    \begin{align}
    	&u[2] \rightarrow \dfrac{\sqrt{5}(4 - 32\mathrm{i}) a_1 c_1}{2\sqrt{13(c_1^2 + c_2^2 + c_3^2)}} \operatorname{sech}\left[ 2a_1 c - \dfrac{1}{2} \ln (c_1^2 + c_2^2 + c_3^2) -\dfrac{1}{2} \ln \left( \dfrac{13}{5} \right) \right] e^{32\mathrm{i} a_1^3 t - 2\mathrm{i} a_1 c}, \\
    	&v[2] \rightarrow \dfrac{\sqrt{5}(4 - 32\mathrm{i}) a_1 c_2}{2\sqrt{13(c_1^2 + c_2^2 + c_3^2)}} \operatorname{sech}\left[ 2a_1 c - \dfrac{1}{2} \ln (c_1^2 + c_2^2 + c_3^2) -\dfrac{1}{2} \ln \left( \dfrac{13}{5} \right) \right] e^{32\mathrm{i} a_1^3 t - 2\mathrm{i} a_1 c}, \\
    	&w[2] \rightarrow \dfrac{\sqrt{5}(4 - 32\mathrm{i}) a_1 c_3}{2\sqrt{13(c_1^2 + c_2^2 + c_3^2)}} \operatorname{sech}\left[ 2a_1 c - \dfrac{1}{2} \ln (c_1^2 + c_2^2 + c_3^2) -\dfrac{1}{2} \ln \left( \dfrac{13}{5} \right) \right] e^{32\mathrm{i} a_1^3 t - 2\mathrm{i} a_1 c}.
    \end{align}
    
    After interacting($t\rightarrow-\infty$), we can get
    
    \begin{align}
    	u[2] &\rightarrow \dfrac{\sqrt{13}(4 + 32\mathrm{i}) a_1 c_1}{2\sqrt{5(c_1^2 + c_2^2 + c_3^2)}} \operatorname{sech}\left[ 2a_1 c - \dfrac{1}{2} \ln (c_1^2 + c_2^2 + c_3^2) -\dfrac{1}{2} \ln \left( \dfrac{5}{13} \right) \right] e^{32\mathrm{i} a_1^3 t - 2\mathrm{i} a_1 c},\\
    	v[2] &\rightarrow \dfrac{\sqrt{13}(4 + 32\mathrm{i}) a_1 c_2}{2\sqrt{5(c_1^2 + c_2^2 + c_3^2)}} \operatorname{sech}\left[ 2a_1 c - \dfrac{1}{2} \ln (c_1^2 + c_2^2 + c_3^2) -\dfrac{1}{2} \ln \left( \dfrac{5}{13} \right) \right] e^{32\mathrm{i} a_1^3 t - 2\mathrm{i} a_1 c},\\
    	w[2] &\rightarrow \dfrac{\sqrt{13}(4 + 32\mathrm{i}) a_1 c_3}{2\sqrt{5(c_1^2 + c_2^2 + c_3^2)}} \operatorname{sech}\left[ 2a_1 c - \dfrac{1}{2} \ln (c_1^2 + c_2^2 + c_3^2) -\dfrac{1}{2} \ln \left( \dfrac{5}{13} \right) \right] e^{32\mathrm{i} a_1^3 t - 2\mathrm{i} a_1 c}.
    \end{align}

    \subsection{Dark-bright-bright soliton solution}
    \numberwithin{equation}{section}
    In this section, we choose the $u=ae^{\i\sqrt6 ax}, v=0, w=0$ to be the initial solution, and $a$ is an arbitrarily real number. By this we can obtain the following fundamental solutions of \eqref{b1}--\eqref{b2},
    \begin{align}
    	\varPsi_{i} &= 
    	\begin{pmatrix}
    		e^{\frac{\i \sqrt{\scriptstyle 6} ax}{2}} & 0 & 0 & 0 \\
    		0 & e^{-\frac{\i\sqrt{\scriptstyle 6} ax}{2}} & 0 & 0 \\
    		0 & 0 & 1 & 0 \\
    		0 & 0 & 0 & 1
    	\end{pmatrix}
    	\begin{pmatrix}
    		\displaystyle \frac{2a}{k_i + \i\sqrt{6}a + 2\i \zeta_i} e^{A_{i1}} \\[2ex]
    		e^{A_{i1}} \\[1ex]
    		c_4 e^{A_{i2}} \\[1ex]
    		c_3 e^{A_{i2}}
    	\end{pmatrix}, 
    \end{align}
    where $c_3, c_4$ are arbitrarily real numbers, and
    \begin{flalign*}
    	k_i &= \sqrt{-4\sqrt{6}a\zeta_i -10a^2 -4\zeta_i^2}, \\
    	A_{i1} &= \frac{k_i}{2}(\xi_{i1}t + x), \quad A_{i2} = \i\zeta_i(\zeta_i^2t + x), \\
    	\xi_i &= -2\sqrt{6}a\zeta_i + 4a^2 + 4\zeta_i^2.
    \end{flalign*}
    In this case, we have the expression of solutions after $N$-fold DT
    \begin{equation}
    	u[N] = e^{\mathrm{i} \sqrt{6}\,a\,x}
    	\left( 
    	a + 2\mathrm{i} \frac{\det(M_1)}{\det(M)} 
    	\right), 
    	\quad
    	v[N] = 2\mathrm{i} \,
    	e^{\frac{\i \sqrt{\scriptstyle 6} ax}{2}}
    	\frac{\det(M_2)}{\det(M)}, 
    	\quad
    	w[N] = 2\mathrm{i} \,
    	e^{\frac{\i\sqrt{\scriptstyle 6}ax}{2}}
    	\frac{\det(M_3)}{\det(M)},
    \end{equation}
     where
    \begin{align*}
    	M = \bigl( M_{ij} \bigr)_{1 \le i,j \le N},\quad
    	M_1 = \begin{pmatrix}
    		M & Y_1 \\
    		X_1 & 0 
    	\end{pmatrix},\quad 
    	M_2 = \begin{pmatrix}
    		M & Y_2 \\
    		X_1 & 0 
    	\end{pmatrix},\quad 
    	M_3 = \begin{pmatrix}
    		M & Y_3 \\
    		X_1 & 0 
    	\end{pmatrix},
    \end{align*}
    and
    \begin{flalign*}
    	&M_{ij}=\dfrac{1}{\zeta_j - \zeta_i^{*}} \left[ B_{ij} e^{A_{i1}^{*} + A_{j1}} + \left( c_3^2 + c_4^2 \right) e^{A_{i2}^{*} + A_{j2}} \right],&\\
    	&X_1 = \begin{pmatrix}
    		\displaystyle \frac{2a}{k_{\text{i}} + \i\sqrt{6}a + 2i\zeta_{\text{i}}} e^{A_{11}},  
    		&\displaystyle \frac{2a}{k_{\text{i}} + \i\sqrt{6}a + 2i\zeta_{\text{i}}} e^{A_{21}},
    		& \cdots 
    		& \displaystyle \frac{2a}{k_{\text{i}} + \i\sqrt{6}a + 2i\zeta_{\text{i}}} e^{A_{N1}}
    	\end{pmatrix},&\\
    	&Y_1 = \begin{pmatrix}
    		e^{A_{11}^{*}},
    		&e^{A_{21}^{*}},
    		&\cdots 
    		e^{A_{N1}^{*}}
    	\end{pmatrix}^{T},&\\
    	&Y_2 = \begin{pmatrix}
    		c_4 e^{A_{12} ^{*}},
    		&c_4 e^{A_{22} ^{*}}, 
    		&\cdots 
    		&c_4 e^{A_{N2} ^{*}}
    	\end{pmatrix}^{T},&\\
    	&Y_3 = \begin{pmatrix}
    		c_3 e^{A_{12} ^{*}},
    		&c_3 e^{A_{22} ^{*}}, 
    		&\cdots 
    		&c_3 e^{A_{N2} ^{*}}
    	\end{pmatrix}^{T},&\\
    	&B_{ij} = 1 + \dfrac{4a^2}{(k_i^* - \mathrm{i}\sqrt{6}a - 2\mathrm{i}\zeta_i^*)(k_j + \mathrm{i}\sqrt{6}a + 2\mathrm{i}\zeta_j)}.
    \end{flalign*}
    For $N$=1, the solutions can be written as:
   \begin{align}
   	   u[1] &=ae^{\i\sqrt{6}ax} \dfrac{\sqrt{B}
   	   	\cosh\left[ \alpha + \dfrac{1}{2}\ln C \left( c_{3}^{2} + c_{4}^{2} \right) - \dfrac{1}{2}\ln B \right]
   	   }{\sqrt{C+4a^2}
   	   	\cosh\left[ \alpha + \dfrac{1}{2}\ln C \left( c_{3}^{2} + c_{4}^{2} \right) - \dfrac{1}{2}\ln\left( C+4a^2 \right) \right]
   	   },\\
   	   v[1] &= \frac{
   	   	2\i a c_4 (\zeta_1 - \zeta_1^*)
   	   	\left( \i\sqrt{6a} + 2\i\zeta_1^* - k_1^* \right)
   	   	\exp\left( \frac{\i\sqrt{6} a x}{2} \right)
   	   }{
   	   	e^{\i\beta} \sqrt{C (c_3^2 + c_4^2)} \sqrt{C + 4a^2}
   	   } \quad \mathrm{sech}\left[
   	   \alpha + \frac{1}{2} \ln\left( C (c_3^2 + c_4^2) \right) - \frac{1}{2} \ln\left( C + 4a^2 \right)
   	   \right],\\
   	   w[1] &= \frac{
   	   	2\i a c_3 (\zeta_1 - \zeta_1^*)
   	   	\left( \i\sqrt{6a} + 2\i\zeta_1^* - k_1^* \right)
   	   	\exp\left( \frac{\i\sqrt{6} a x}{2} \right)
   	   }{
   	   	e^{\i\beta} \sqrt{C (c_3^2 + c_4^2)} \sqrt{C + 4a^2}
   	   } \quad  \mathrm{sech}\left[
   	   \alpha + \frac{1}{2} \ln\left( C (c_3^2 + c_4^2) \right) - \frac{1}{2} \ln\left( C + 4a^2 \right)
   	   \right],
   \end{align}
    where
    \begin{align*}
    	B &= C + 4a^2 + 4 \left( \zeta_{1}^{*} - \zeta_{1} \right) \left( \i k_{1}^{*} + \sqrt{6} a + 2 \zeta_{1}^{*} \right),\\
    	C &= \left( \mathrm{i}\sqrt{6}\mathrm{a} + 2\mathrm{i}\zeta_1 + \mathrm{k}_1 \right) \mathrm{k}_1^* + \mathrm{i} \left( -\sqrt{6}\mathrm{a} - 2\mathrm{i}\zeta_1^* \right) \mathrm{k}_1 + 2\mathrm{a} \left( \zeta_1 + \zeta_1^* \right) \sqrt{6} + 6\mathrm{a}^2 + 4\zeta_1 \zeta_1^*,\\
    	\alpha &= \operatorname{Re}\left( A_{12} - A_{11} \right), \quad
    	\beta  = \operatorname{Im}\left( A_{12} - A_{11} \right).
    \end{align*}
    
    \begin{figure}[tbh]
    	\centering
    	\subfigure[]{\includegraphics[height=4.8cm,width=4.8cm]{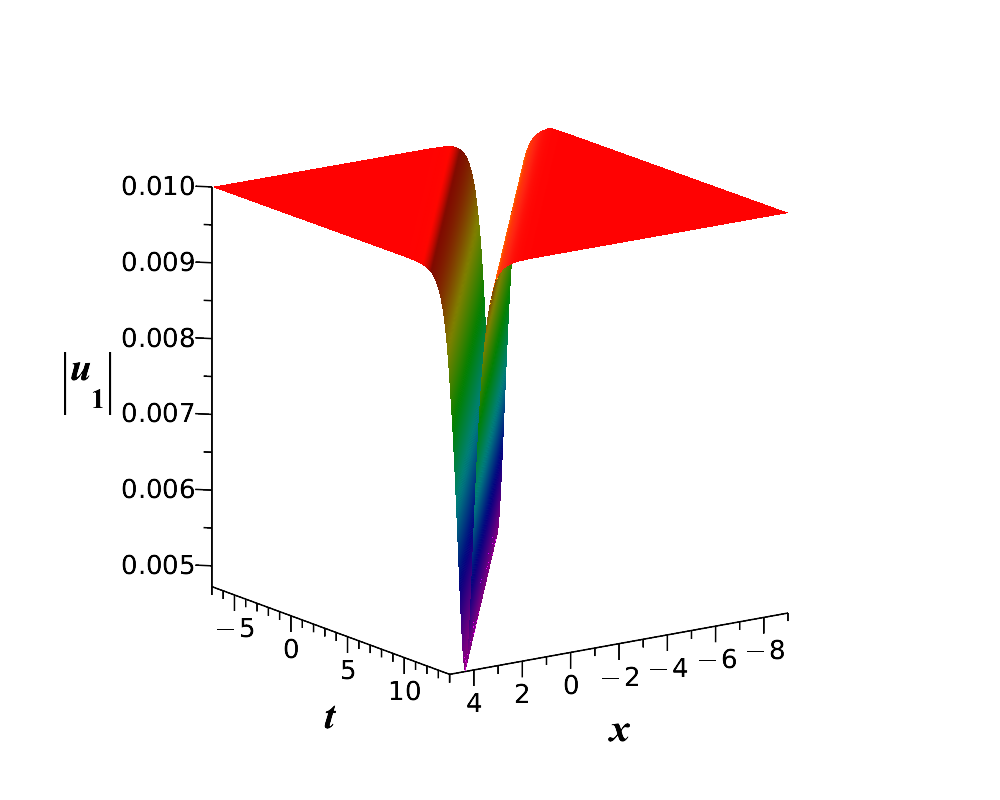}}
    	\subfigure[]{\includegraphics[height=4.8cm,width=4.8cm]{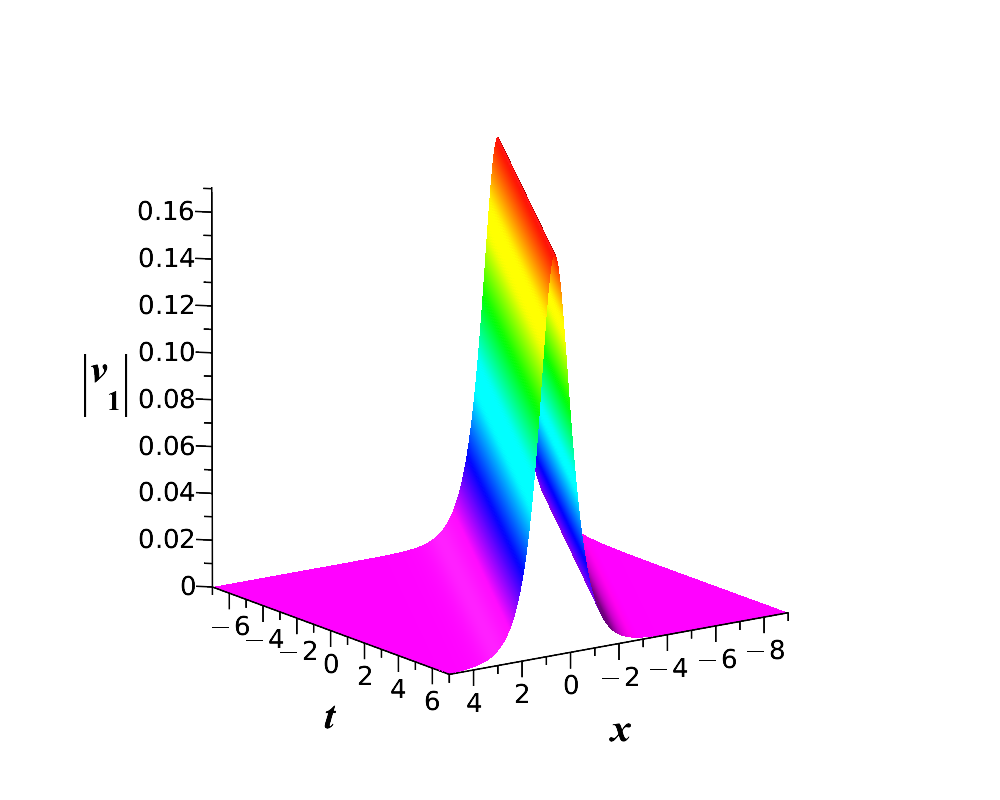}}
    	\subfigure[]{\includegraphics[height=4.8cm,width=4.8cm]{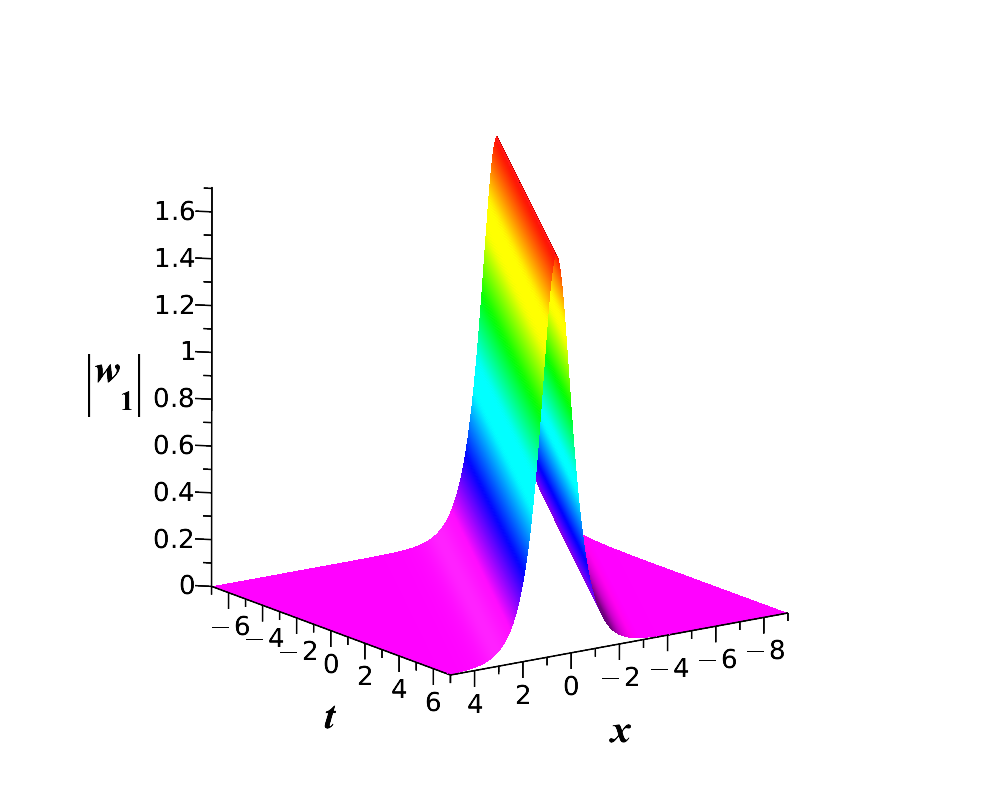}}
    	\caption{The bounded one-dark-bright-bright solution of vcmKdV with $\zeta_1 = \dfrac{1}{2} + \dfrac{1}{4}\mathrm{i}$, $c_1=1, c_2=1, c_3=2$}
    	\label{s2-1}
    \end{figure}
   The one-dark-bright-bright soliton solution are  linear solution, see Fig.\ref{s2-1}. The solution $v[1]$ is still proportional to $w[1]$ with the ratio $c_4 : c_3$. The trajectory of $v[1]$, $w[1]$ is
    \begin{align}
    	Re[i\zeta_{\text{1}}(\zeta_{\text{1}}^2t + x)-\frac{k_{\text{1}}}{2}(\xi_{\text{1}1}t + x) + \frac{1}{2} \ln\left( C (c_3^2 + c_4^2) \right) - \frac{1}{2} \ln\left( C + 4a^2 \right)]=0,
    \end{align}
    and the trajectory of $u[1]$ is
    \begin{align}
    	Re[i\zeta_{\text{1}}(\zeta_{\text{1}}^2t + x)-\frac{k_{\text{1}}}{2}(\xi_{\text{1}1}t + x) + \frac{1}{2} \ln\left( C (c_3^2 + c_4^2) \right) - \frac{1}{2} \ln\left( B \right)]=0,
    \end{align}
    We calculate that
    \begin{align*}
    	|u[1]|_{\min} = 0, \quad
    	|v[1]|_{\max} =  \operatorname{Im}(\zeta_1) \left| 4ac_4 \dfrac{\i\sqrt{6}\,a + 2\i\zeta_1^{*} - k_1^{*}}{\sqrt{C(c_3^2 + c_4^2)} \sqrt{C + 4a^2}} \right|, \quad
    	|w[1]|_{\max} = \operatorname{Im}(\zeta_1) 
    	\left|  4ac_3  \dfrac{	\i\sqrt{6}\,a + 2\i\zeta_1^{*} - k_1^{*}}{\sqrt{C(c_3^2 + c_4^2)} \sqrt{C + 4a^2}} \right|.
    \end{align*}
   For $N$=2, the expression of the solutions after two-fold DT are
    \begin{align*}
    	u[2] &= e^{\displaystyle \i\sqrt{6} a x} \left[ a + 2\i \, \frac{\phi_1 \bigl( M_{12}\varphi_2^* - M_{22}\varphi_1^* \bigr) - \phi_2 \bigl( M_{11}\varphi_2^* - M_{21}\varphi_1^* \bigr)}{M_{11}M_{22} - M_{12}M_{21}} \right], \\
    	v[2] &= 2\i e^{\displaystyle \frac{\i\sqrt{6}ax}{2}} \cdot \frac{\phi_1 \bigl( M_{12}\psi_2^* - M_{22}\psi_1^* \bigr) - \phi_2 \bigl( M_{11}\psi_2^* - M_{21}\psi_1^* \bigr)}{M_{11}M_{22} - M_{12}M_{21}}, \\
    	w[2] &= 2\i e^{\displaystyle \frac{\i\sqrt{6} a x}{2}} \cdot \frac{\phi_1 \bigl( M_{12}\chi_2^* - M_{22}\chi_1^* \bigr) - \phi_2 \bigl( M_{11}\chi_2^* - M_{21}\chi_1^* \bigr)}{M_{11}M_{22} - M_{12}M_{21}}.
    \end{align*}
    We present the bounded two-dark-bright-bright soliton solution with the concrete parameters, see Fig.\ref{s2-2}. The solutions after two fold DT are two solitons with a periodic solution similar to bright-bright-bright soliton solution. And with the parameters as $\zeta_1 = -\zeta_2$ or $\zeta_1 = -\zeta_2^*$, we can also obtain the breather-like solutions. 
    \begin{figure}[tbh]
    	\centering
    	\subfigure[]{\includegraphics[height=4.8cm,width=4.8cm]{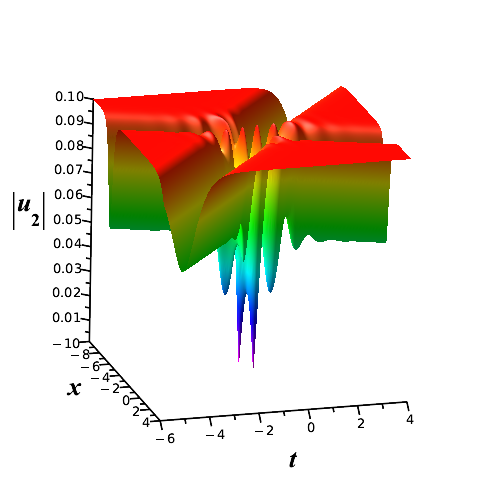}}
    	\subfigure[]{\includegraphics[height=4.8cm,width=4.8cm]{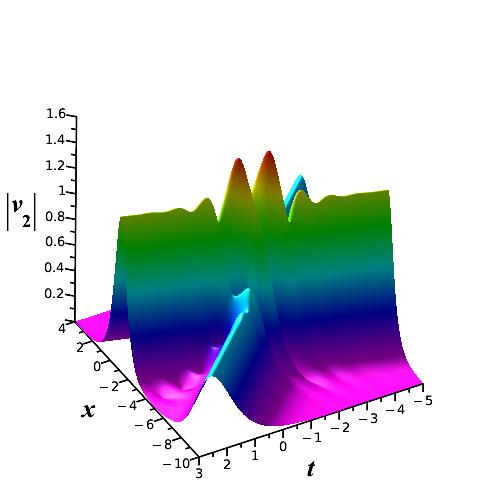}}
    	\subfigure[]{\includegraphics[height=4.8cm,width=4.8cm]{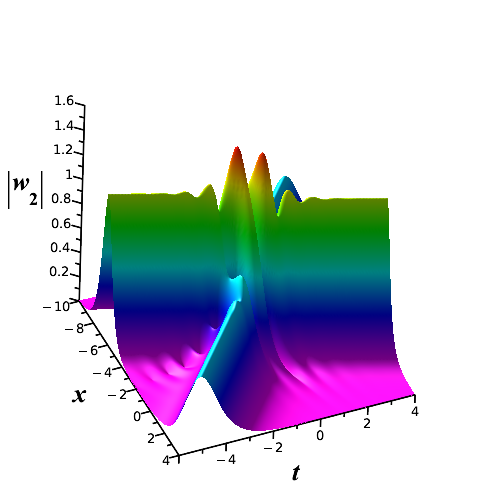}}
    	\subfigure[]{\includegraphics[height=4.8cm,width=4.5cm]{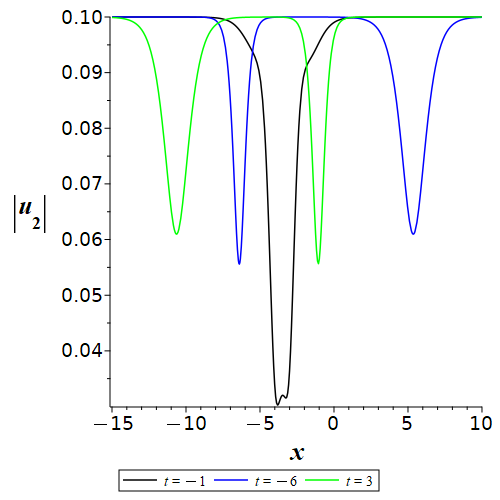}}
    	\subfigure[]{\includegraphics[height=4.8cm,width=4.5cm]{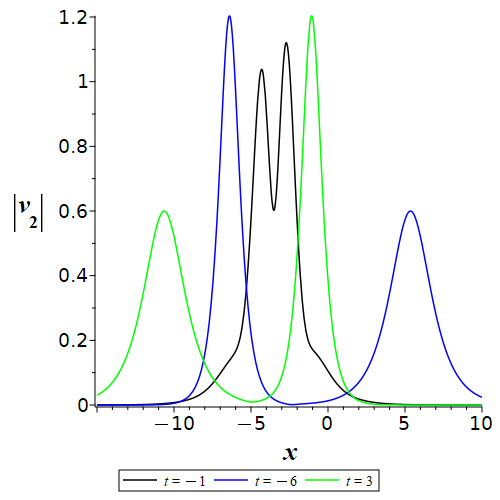}}
    	\subfigure[]{\includegraphics[height=4.8cm,width=4.5cm]{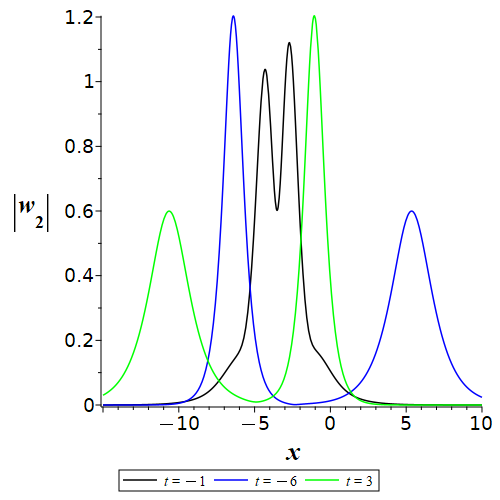}}
    	\caption{Panel (a), (b), (c) are the 3D plots of the bounded two-dark-bright-bright solution of vcmKdV with $\zeta_1 = \dfrac{4}{9} + \dfrac{23}{27}\i$, $\zeta_2 = -\dfrac{4}{9} - \dfrac{23}{54}\i$,$c_1 = \dfrac{1}{10}$, $c_2 = \dfrac{1}{10}$, $a= \dfrac{1}{10}$. Panel (d), (e), (f) are the profiles of the bounded two-bright-bright-bright solution at $t=0, t=-6, t=4$.}
    	\label{s2-2}
    \end{figure}

    \section{Breather solution}\label{3}
    \numberwithin{equation}{section}
    
    In this section, we choose the solutions
    \begin{equation}\label{d1}
    u = a e^{\i \sqrt{6(a^2 + a_1^2 + a_2^2)} x},\quad
    v = a_1 e^{\i \sqrt{6(a^2 + a_1^2 + a_2^2)} x},\quad
    w = a_2 e^{\i \sqrt{6(a^2 + a_1^2 + a_2^2)} x},
    \end{equation}
    as the initial solutions, where $a$, $a_1$ and $a_2$ are arbitrary parameters. Therefore, substituting these into \eqref{b1} and \eqref{b2}, we obtain
    \begin{align}
    	\varPsi_j = 
    	{\scriptstyle 
    		\begin{pmatrix}
    			 e^{\textstyle \frac{\sqrt{6}}{2} \i dx} & 0 & 0 & 0 \\
    			0 & e^{\textstyle -\frac{\sqrt{6}}{2} \i dx} & 0 & 0 \\
    			0 & 0 & e^{\textstyle -\frac{\sqrt{6}}{2} \i dx} & 0 \\
    			0 & 0 & 0 & e^{\textstyle -\frac{\sqrt{6}}{2} \i dx}
    		\end{pmatrix}}
    	\begin{pmatrix}
    		\frac{1}{2} \left[B_{j1} e^{A_{j1}} - B_{j2} e^{-A_{j1}} \right] \\
    		a \left( e^{A_{j1}} - e^{-A_{j1}} \right) - \left( a_1\alpha_1 + a_2\alpha_2 \right) e^{A_{j2}} \\
    		a_2 \left( e^{A_{j1}} - e^{-A_{j1}} \right) + a\alpha_1 e^{A_{j2}} \\
    		a_1 \left( e^{A_{j1}} - e^{-A_{j1}} \right) + a\alpha_2 e^{A_{j2}}
    	\end{pmatrix},
    \end{align}
    where
    \begin{align}
    	&d = \sqrt{a^2 + a_1^2 + a_2^2}, \quad 
    	k_j = \sqrt{-4\sqrt{6}\, d\,\zeta_j - 10d^2 - 4\zeta_j^2},\\
    	&B_{j1} =  \i\sqrt{6}d + 2\i\zeta_j + k_j ,\quad
    	B_{j2} =  \i\sqrt{6}d + 2\i\zeta_j - k_j ,\\
    	&A_{j1} = \dfrac{k_j}{2} (\xi_{j1}\, t - x) + s_j, \quad 
    	A_{j2} = 4\i\zeta_j^3 t + \i\zeta_j x + \dfrac{\sqrt{6}\, d}{2} \i x,\\
    	&\xi_{j1} = 2\sqrt{6}\, d\,\zeta_j - 4d^2 - 4\zeta_j^2,
    \end{align}
    and the $s_i = s_{iR} + \i s_{iI} (i = 1, \dots, N)$. To simplify the form of $k_j$, we choose the eignvalue $\zeta_j$ as $\zeta_j = d \sinh(p_j) - \frac{\sqrt{6}}{2} d,p_j = p_{jR} + \i p_{jI}$, so we can get
    \begin{align*}
    	k_j &= 2\mathrm{i} d \cosh(p_j),\quad\xi_{j1} = d^2 \left[ -4 \sinh^2(p_j) + 6\sqrt{6} \sinh(p_j) - 16 \right].
    \end{align*} 
    After all, the $N$-th solution can be written as
    \begin{equation}
    	u[N] = e^{\mathrm{i}\sqrt{6}\,a\,x}
    	\left( 
    	a + 2\mathrm{i} \frac{\det(M_1)}{\det(M)} 
    	\right), 
    	\quad
    	v[N] = e^{\mathrm{i}\sqrt{6}\,a\,x}
    	\left( 
    	a_1 + 2\mathrm{i} \frac{\det(M_2)}{\det(M)} 
    	\right), 
    	\quad
    	w[N] = e^{\mathrm{i}\sqrt{6}\,a\,x}
    	\left( 
    	a_2 + 2\mathrm{i} \frac{\det(M_3)}{\det(M)} 
    	\right),
    \end{equation}
    where
    \begin{align*}
    	M = \bigl( M_{ij} \bigr)_{1 \le i,j \le N},\quad
    	M_1 = \begin{pmatrix}
    		M & Y_1 \\
    		X_1 & 0 
    	\end{pmatrix},\quad 
    	M_2 = \begin{pmatrix}
    		M & Y_2 \\
    		X_1 & 0 
    	\end{pmatrix},\quad 
    	M_3 = \begin{pmatrix}
    		M & Y_3 \\
    		X_1 & 0 
    	\end{pmatrix},
    \end{align*}
    and
    \begin{align*}
    	M_{ij} &= \frac{1}{4(\zeta_j - \zeta_i^*)} \Bigl[ (B_{j1}B_{i1}^* + 4d^2) e^{A_{j1} + A_{i1}^*}
    	 - (B_{j2}B_{i1}^* + 4d^2) e^{-A_{j1} + A_{i1}^*} - (B_{j1}B_{i2}^* + 4d^2) e^{A_{j1} - A_{i1}^*} \\
    	 &+ (B_{j2}B_{i2}^* + 4d^2) e^{-A_{j1} - A_{i1}^*} + 4\bigl[ (\alpha_1^2 + \alpha_2^2)d^2 - (a_1\alpha_2 - a_2\alpha_1) \bigr] e^{A_{j2} + A_{i2}^*} \Bigr],\\
    	X_1 &= \begin{pmatrix}
    		\dfrac{1}{2} \left[ B_{11} e^{A_{11}} - B_{12} e^{-A_{11}} \right] &
    		\dfrac{1}{2} \left[ B_{21} e^{A_{21}} - B_{22} e^{-A_{21}} \right] &
    		\cdots &
    		\dfrac{1}{2} \left[ B_{N1} e^{A_{N1}} - B_{N2} e^{-A_{N1}} \right]
    	\end{pmatrix} ,\\
    	Y_1 &= \begin{pmatrix}
    		a \left( e^{A_{11}^*} - e^{-A_{11}^*} \right) - (a_1\alpha_1 + a_2\alpha_2) e^{A_{12}^*} &
    		\cdots &
    		a \left( e^{A_{N1}^*} - e^{-A_{N1}^*} \right) - (a_1\alpha_1 + a_2\alpha_2) e^{A_{N2}^*}
    	\end{pmatrix}^T ,\\
    	Y_2 &= \begin{pmatrix}
    		a_1 \left( e^{A_{11}^*} - e^{-A_{11}^*} \right) - a\alpha_1 e^{A_{12}^*} &
    		a_1 \left( e^{A_{21}^*} - e^{-A_{21}^*} \right) - a\alpha_1 e^{A_{22}^*} &
    		\cdots &
    		a_1 \left( e^{A_{N1}^*} - e^{-A_{N1}^*} \right) - a\alpha_1 e^{A_{N2}^*}
    	\end{pmatrix}^T ,\\
    	Y_3 &= \begin{pmatrix}
    		a_2 \left( \mathrm{e}^{A_{11}^*} - \mathrm{e}^{-A_{11}^*} \right) - a\alpha_2 \mathrm{e}^{A_{12}^*} &
    		a_2 \left( \mathrm{e}^{A_{21}^*} - \mathrm{e}^{-A_{21}^*} \right) - a\alpha_2 \mathrm{e}^{A_{22}^*} &
    		\cdots &
    		a_2 \left( \mathrm{e}^{A_{N1}^*} - \mathrm{e}^{-A_{N1}^*} \right) - a\alpha_2 \mathrm{e}^{A_{N2}^*}
    	\end{pmatrix}^T.
    \end{align*}
     
     \begin{figure}[tbh]
     	\centering
     	\subfigure[]{\includegraphics[height=4.5cm,width=4.9cm]{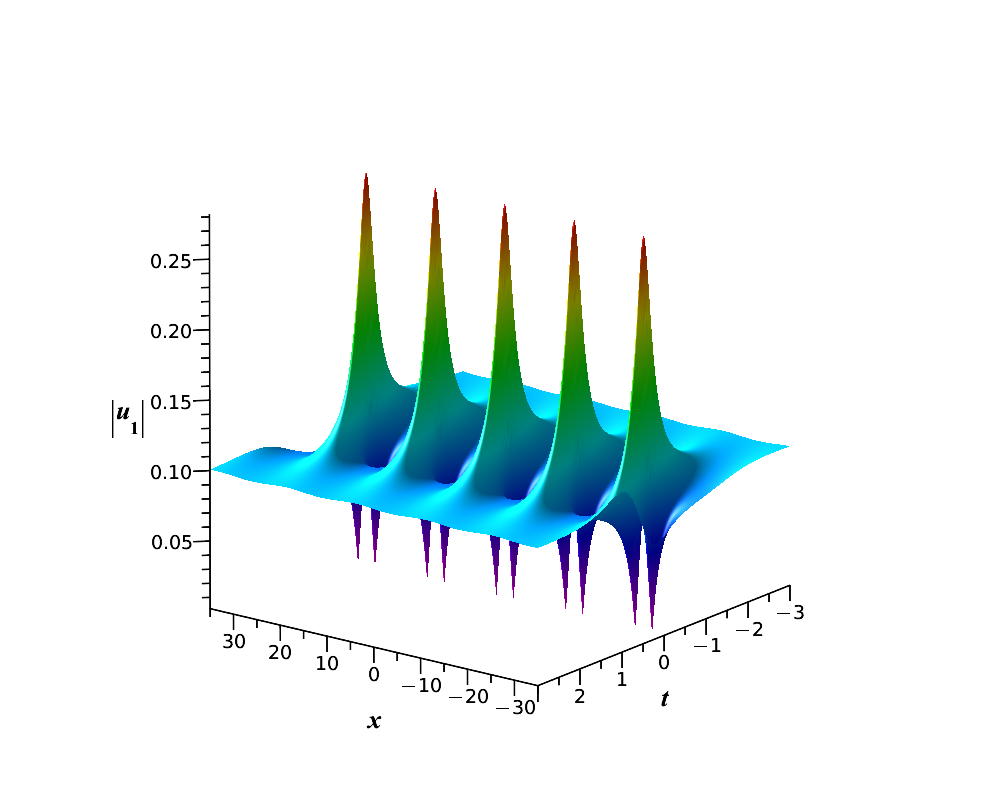}}
     	\subfigure[]{\includegraphics[height=4.5cm,width=4.9cm]{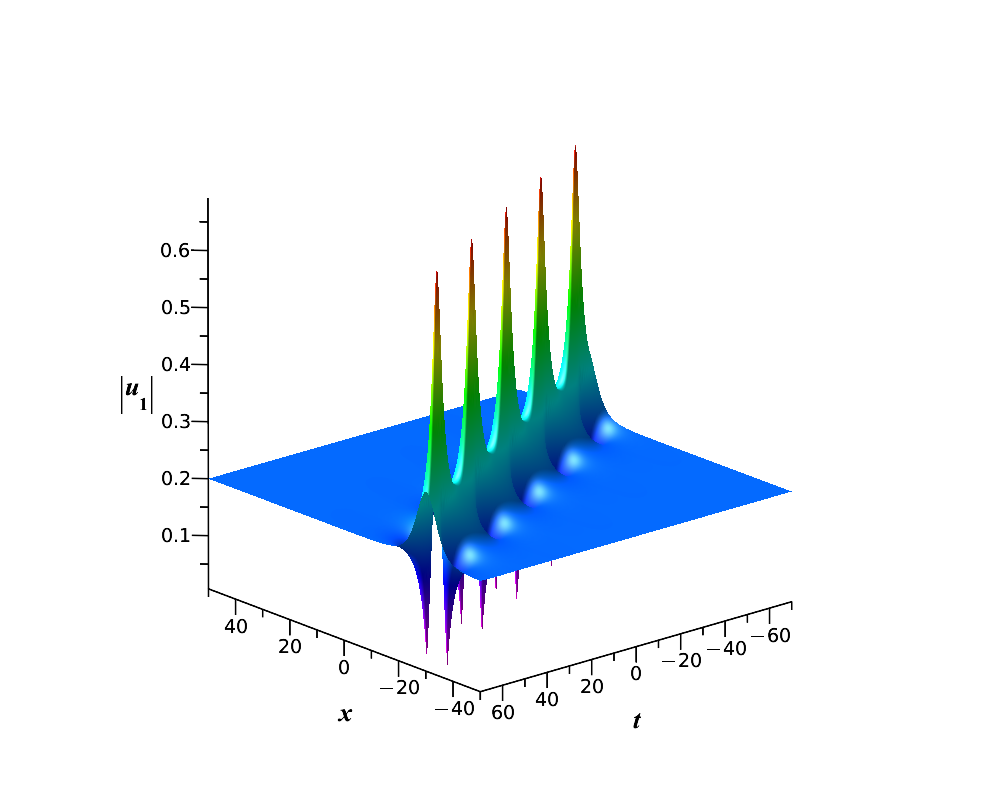}}
     	\subfigure[]{\includegraphics[height=4.5cm,width=4.9cm]{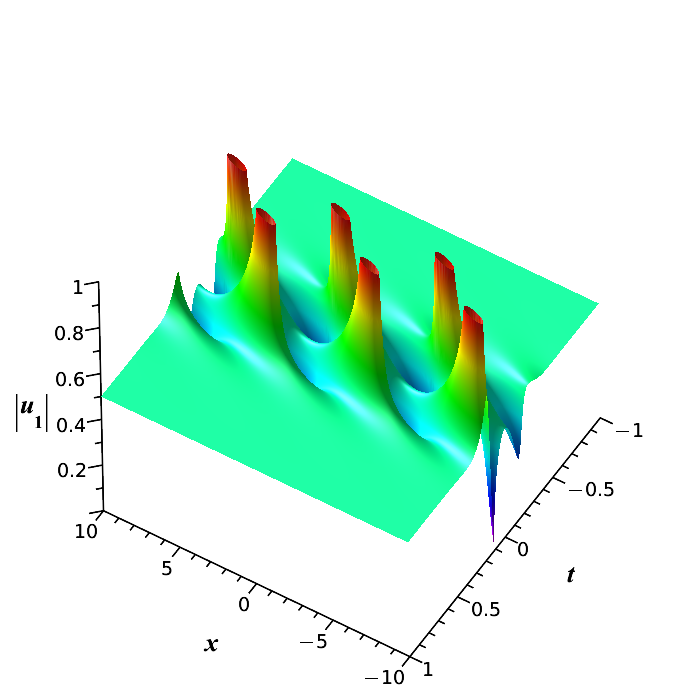}}
     	\caption{Panel (a) is the bounded Akhmediev breather with $p_1=2\i,s_1=0,\alpha_1=0,\alpha_2=0,a_1=0,a_2=\dfrac{1}{2},a=\dfrac{1}{10}$, (b) is the bounded regular breather with $p_1=\dfrac{\pi}{2}\i+\dfrac{2}{3},s_1=0,\alpha_1=0,\alpha_2=0,a_1=0,a_2=0,a=\dfrac{1}{5}$ , and (c) is the collapsing Akhmediev breather with $p_1=2\i,s_1=1,\alpha_1=0,\alpha_2=0,a_1=0,a_2=1,a=\dfrac{1}{2}$}
     	\label{b-1}
     \end{figure} 
    Specifically, the general one-breather solution can be given as
    \begin{align}
    	u[1] &= e^{\mathrm{i}\sqrt{6} d x} \left[ 
    	a \frac{F_1 + F_2}{F_1} + 
    	4 (a_1\alpha_2 + a_2\alpha_1) (\zeta_1 - \zeta_1^*) \frac{F_3}{F_1} 
    	\right],\label{d8} \\ 
    	v[1] &= e^{\mathrm{i}\sqrt{6} d x} \left[ 
    	a_1 \frac{F_1 + F_2}{F_1} - 
    	4a\alpha_2  (\zeta_1 - \zeta_1^*) \frac{F_3}{F_1} 
    	\right],\label{d9} \\ 
    	w[1] &= e^{\mathrm{i}\sqrt{6} d x} \left[ 
    	a_2 \frac{F_1 + F_2}{F_1} - 
    	4a\alpha_1  (\zeta_1 - \zeta_1^*) \frac{F_3}{F_1} 
    	\right],\label{d10}
    \end{align}
    where
    \begin{align*}
    	F_1 &= \bigl[ k_{1}^{*}k_1 + 2\sqrt{6}d(\zeta_1 + \zeta_1^*) + 4\zeta_1\zeta_1^* + 10d^2 \bigr] \cosh(A_{11} + A_{11}^*) \\
    	&\quad - \bigl[ \mathrm{i}(\sqrt{6}d + 2\zeta_1^*)k_1 - \mathrm{i}(\sqrt{6}d + 2\zeta_1)k_1^* \bigr] \sinh(A_{11} + A_{11}^*) \\
    	&\quad + \bigl[ k_{1}^{*}k_1 - 2\sqrt{6}d(\zeta_1 + \zeta_1^*) - 4\zeta_1\zeta_1^* - 10d^2 \bigr] \cosh(A_{11} - A_{11}^*) \\
    	&\quad - \bigl[ \mathrm{i}(\sqrt{6}d + 2\zeta_1^*)k_1 + \mathrm{i}(\sqrt{6}d + 2\zeta_1)k_1^* \bigr] \sinh(A_{11} - A_{11}^*) \\
    	&\quad + 2\left[ \left( \alpha _{1}^{2}+\alpha _{2}^{2} \right) d^2-\left( a_1\alpha _1-a_2\alpha _2 \right) ^2 \right] e^{A_{12}+A_{12}^{*}}\\
    	F_2 &= -\bigl[ (4\mathrm{i}\zeta_1 - 4\mathrm{i}\zeta_1^*) k_1 \bigr] \sinh(A_{11} + A_{11}^*) \\
    	&\quad + \bigl[ 4\sqrt{6}d(-3\zeta_1 - 5\zeta_1^*) + 8\zeta_1^2 - 8\zeta_1\zeta_1^* \bigr] \cosh(A_{11} + A_{11}^*) \\
    	&\quad + \bigl[ (4\mathrm{i}\zeta_1 - 4\mathrm{i}\zeta_1^*) k_1 \bigr] \sinh(A_{11} - A_{11}^*) \\
    	&\quad + \bigl[ -4\sqrt{6}d(-3\zeta_1 - 5\zeta_1^*) - 8\zeta_1^2 + 8\zeta_1\zeta_1^* \bigr] \cosh(A_{11} + A_{11}^*)\\
    	F_3 &= \left[ \i k_1 \cosh(A_{11}) - (d \sqrt{6} + 2 \zeta_1) \sinh(A_{11}) \right] \exp\left(-4\i (\zeta_1^*)^3 t - \i \zeta_1^* x - \frac{\i \sqrt{6}}{2} d x\right).
    \end{align*}
    \begin{figure}[tbh]
     	\centering
     	\subfigure[]{\includegraphics[height=4.5cm,width=6cm]{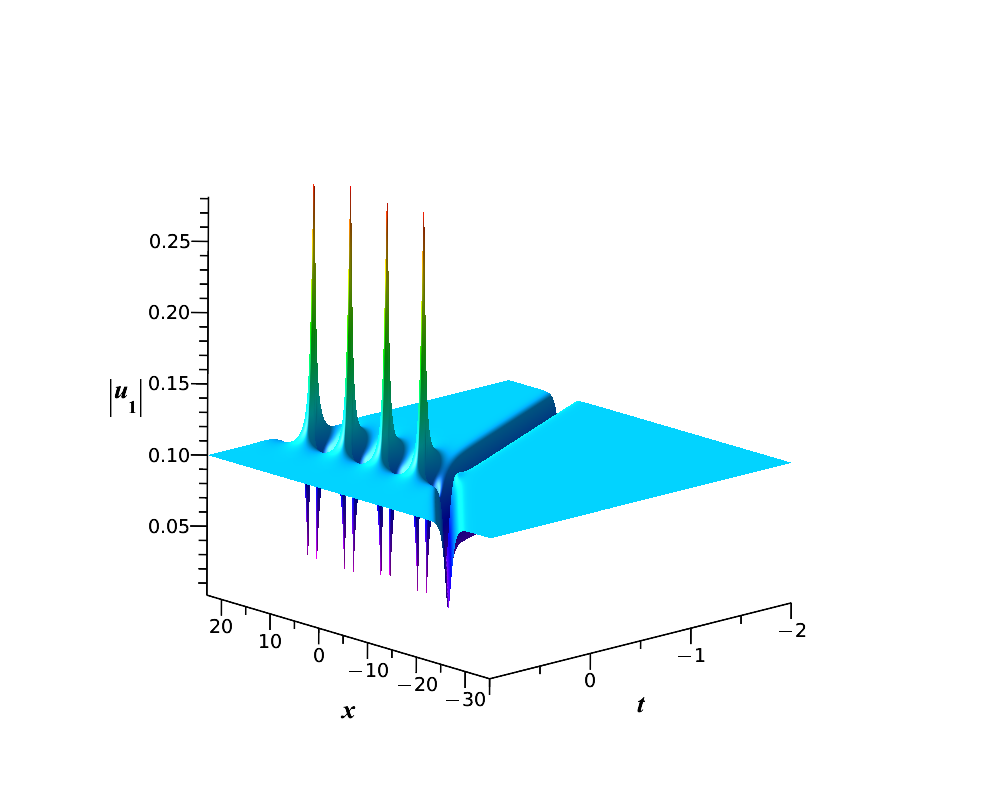}}
     	\subfigure[]{\includegraphics[height=4.5cm,width=6cm]{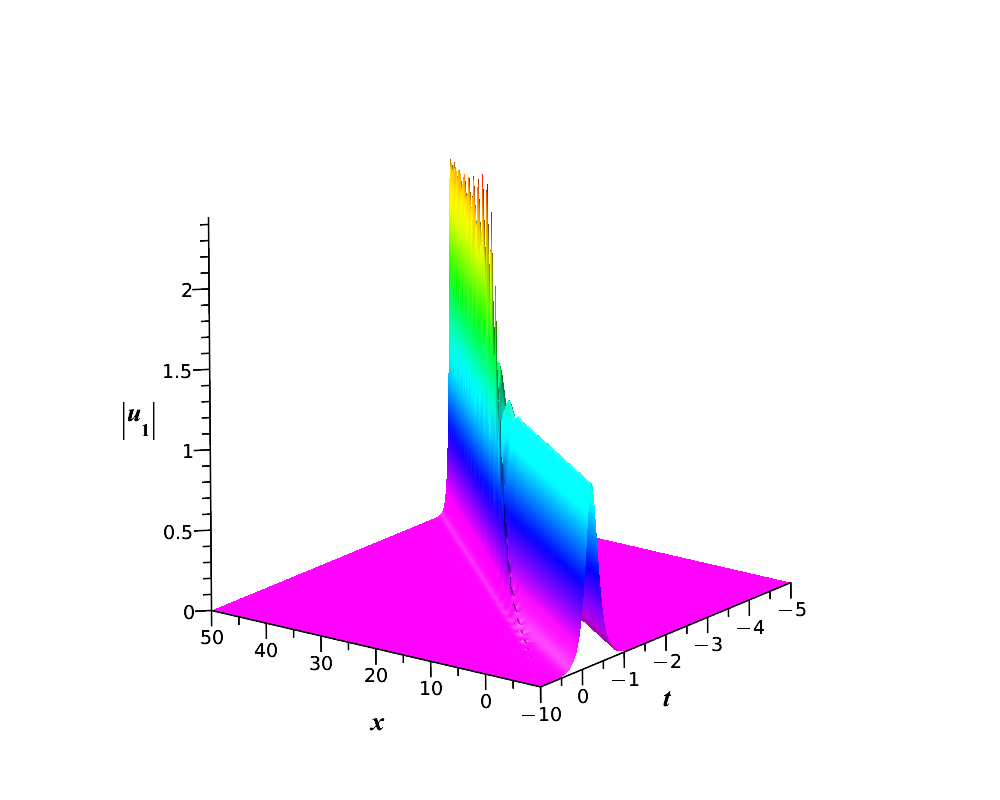}}
     	\caption{Panel (a) is the bounded Akhmediev breather interacting with two bright solitons with $p_1=2\i, s_1=0, \alpha_1=0, \alpha_2=\frac{1}{10}, a_1=0, a_2=1, a=\dfrac{1}{10}$. Panel(b) is a bounded breather wave interacting with two bright solitons with $p_1=\dfrac{\pi}{2}\i+\dfrac{2}{3}, s_1=0, \alpha_1=\frac{1}{50000},  \alpha_2=\frac{1}{1000}, a_1=1, a_2=\frac{1}{10}, a=0$.}
     	\label{b-2}
    \end{figure}
    \hfill
    \begin{figure}[tbh]
    	\centering
    	\subfigure[]{\includegraphics[height=4.5cm,width=5cm]{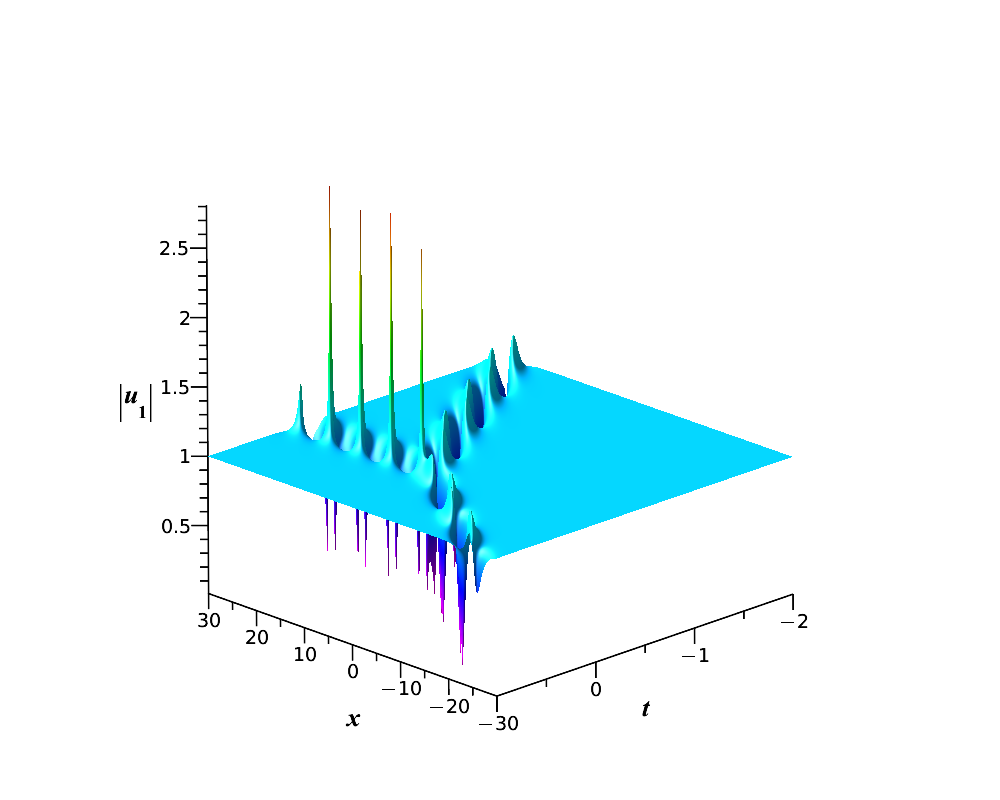}}
    	\subfigure[]{\includegraphics[height=4.5cm,width=5cm]{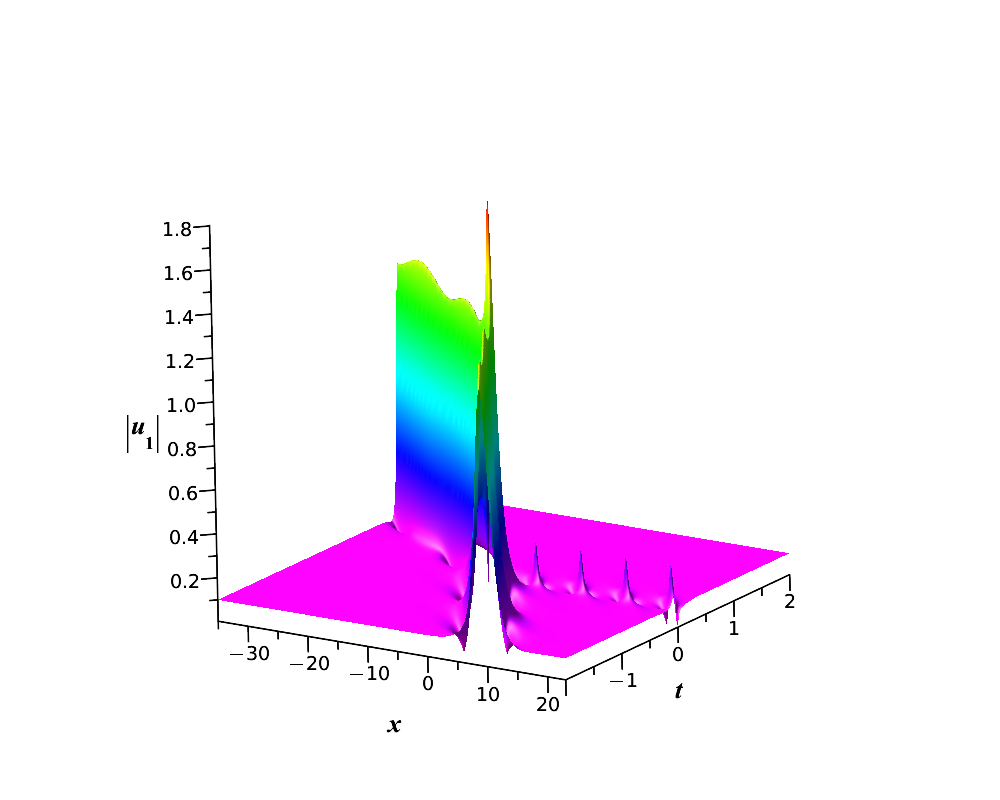}}
    	\subfigure[]{\includegraphics[height=4.5cm,width=5cm]{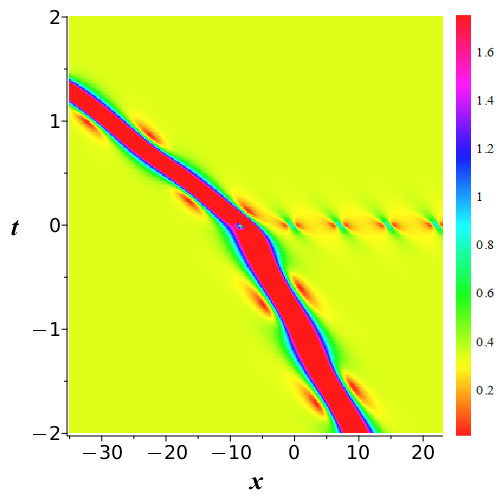}}
    	\caption{Panel (a) is the bounded Akhmediev breather interacting with two breather-breather solutions with $p_1=2\i, s_1=0, \alpha_1=150, \alpha_2=11, a_1=\frac{1}{2},a_2=\frac{2}{5}, a=1$. Panel(b) is the bounded Akhmediev breather interacting with two soliton-breather-breather solutions with $p_1=2\i, s_1=0, \alpha_1=\frac{1}{50000}, \alpha_2=\frac{1}{1000}, a_1=1, a_2=0, a=\frac{1}{10}$, and (c) is the density plot of this solution (b).}
    	\label{b-3}
    \end{figure}
    The behavior of the one-brether solution is governed by multiple parameters $\alpha_{i}, a_{i}, a (i=1,2)$ and the eigenvalue. To characterize these dependencies, we choose the component $u$ and delineate our analysis into the following cases
    
    (1)If $\alpha_1=0, \alpha_2=0$,we can obtain the Akhmediev breather and the regular one-breather solution, we can see that the three components have the ratio $u\left[ 1 \right] :v\left[ 1 \right] :w\left[ 1 \right] =c_1:c_2:c_3$. If we choose $p_{1R}=0$ and $s_1=0$, we can get the bounded Akhmediev breather. If we choose $p_{1R} \ne 0$, we can obtain the regular one-breather solution. It is worth to point out that if we choose $p_{1R}=0$ and $s_{1R}\ne0$, we can obtain the collapsing Akhmediev breather, see Fig.\ref{b-1}.
    
    (2)If $\alpha_1$ and $\alpha_2$ are not both zero, we can obtain the dark-bright solitons interact with breather solution, see Fig.\ref{b-2}. If $a=0, \left( a_1\alpha _2+a_2\alpha _1 \right) \ne 0$, we can obtain a breather interacting with two bright solitons. If $a \ne 0,\left( a_1\alpha _2+a_2\alpha _1 \right)=0$ and we choose $p_{1R}=0$, the Akhmediev breather interacting with two dark solitons can be presented.
    
    (3)If $\alpha_1$ and $\alpha_2$ are not both zero, $a \ne 0$ , $\left( a_1\alpha _2+a_2\alpha _1 \right) \ne 0$ and with $p_{1R}=0$, we can obtain the Akhmediev breather interacting with two breather-breather waves, see Fig.\ref{b-3}(a). By adjusting the parameters, we can demonstrate a special solution which has the Akhmediev breather interacting with two soliton-breather-breather solution. see Fig.\ref{b-3}(b). If we choose $p_{1R} \ne 0$,we can get the regular breather waves interacting with two other waves, and it is worth to show that a soliton interacting with two breathers can also be plotted,see Fig.\ref{b-4}.
    
    Similar to component $u$, the components $v$ and $w$ can derive the same characters based on equations \eqref{d8}-\eqref{d10}.
    
    \begin{figure}[tbh]
    	\centering
    	\subfigure[]{\includegraphics[height=4.5cm,width=6cm]{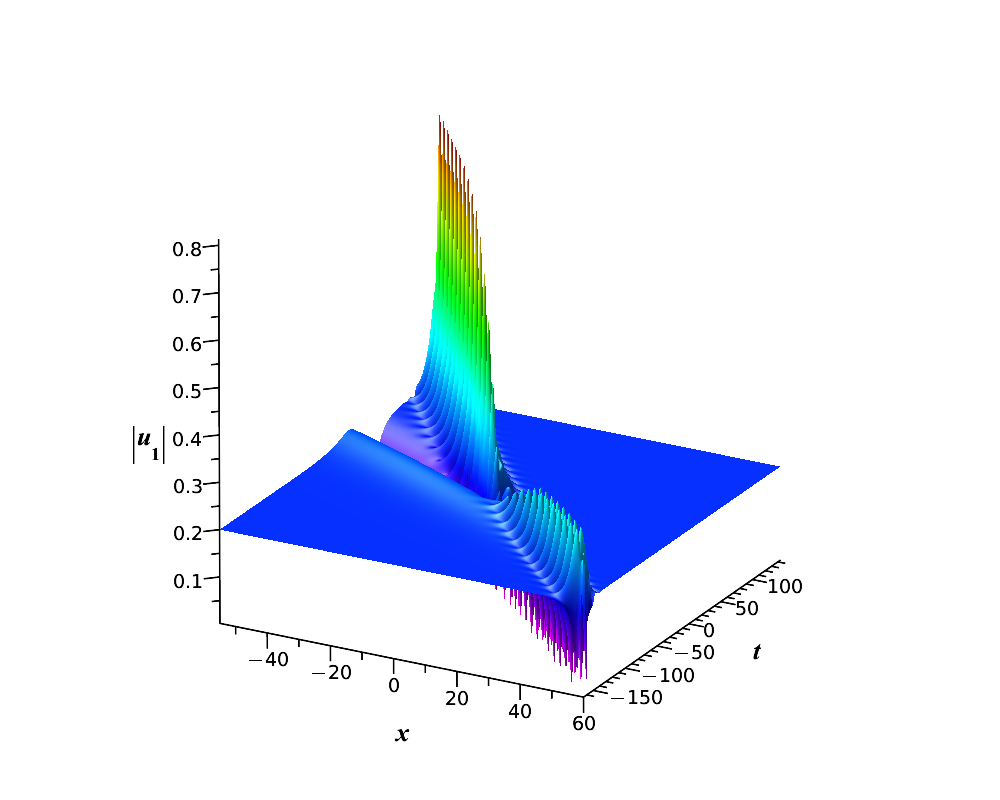}}
    	\subfigure[]{\includegraphics[height=4.5cm,width=6cm]{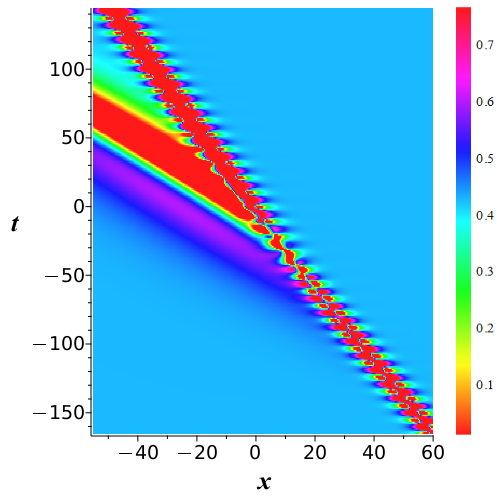}}
    	\caption{Panel(a) is the solution of two breather waves interacting with a dark soliton solution with $p_1=\dfrac{\pi}{2}\i+1, s_1=0,\alpha_1=\frac{1}{10},\alpha_2=\frac{1}{2},a_1=\frac{1}{10},a_2=0,a=\frac{1}{5}$, and panel(b) is the density plot of this solution (b).}
    	\label{b-4}
    \end{figure}   
    
    \section{Positons and rogue wave solution of vcmKdV equation}\label{4}
    \numberwithin{equation}{section}
    \subsection{Positons}
    In this section, we consider the zero "seed" solution \eqref{c1} with $\zeta _1=ia+i\epsilon $, where $\varepsilon$ is a small parameter. Then we can get
    \begin{align}
    	\Psi_1 = \begin{pmatrix}
    		e^{-8\i \Big( \left( -\frac{\varepsilon^2}{2} + (\i b - a)\varepsilon + \i a b - \frac{a^2}{2} + \frac{b^2}{2} \right)t + \frac{x}{8} \Big)(b + \i a + \i\varepsilon)} \\
    		c_1 e^{8\i \Big( \left( -\frac{\varepsilon^2}{2} + (\i b - a)\varepsilon + \i a b - \frac{a^2}{2} + \frac{b^2}{2} \right)t + \frac{x}{8} \Big)(b + \i a + \i\varepsilon)} \\
    		c_2 e^{8\i \Big( \left( -\frac{\varepsilon^2}{2} + (\i b - a)\varepsilon + \i ab - \frac{a^2}{2} + \frac{b^2}{2} \right)t + \frac{x}{8} \Big)(b + \i a + \i\varepsilon)} \\
    		c_3 e^{8\i \Big( \left( -\frac{\varepsilon^2}{2} + (\i b - a)\varepsilon + \i ab - \frac{a^2}{2} + \frac{b^2}{2} \right)t + \frac{x}{8} \Big)(b + \i a + \i\varepsilon)}
    	\end{pmatrix}.
    \end{align}
    By setting eigenvalues 
    $\zeta_i = \zeta_1 + \varepsilon, 
    \zeta_i^{*} = \zeta_1^{*} + \widetilde{\varepsilon}$, $(i = 2, 3, \dots, N)$ , and 
    $\varepsilon, \widetilde{\varepsilon}$ ~are different small parameters. 
    With the higher-order Taylor expansion and limit technique, we obtain the degenerate form of the solutions. 
    By substituting 
    $\varPsi_i = \varPsi_1(\zeta_1 + \varepsilon)$, 
    $(i = 2, 3, \dots, N)$, we can derive the $N$-positon solutions of vcmKdV equation based on the zero ``seed'' solution.
    \begin{equation}\label{e2}
    	u[N] = 2\mathrm{i} \frac{\det(M_1)}{\det(M)},\quad
    	v[N] = 2\mathrm{i} \frac{\det(M_2)}{\det(M)},\quad
    	w[N] = 2\mathrm{i} \frac{\det(M_3)}{\det(M)},
    \end{equation}
    where
    \begin{align}
    	M &= \lim_{\varepsilon \to 0, \widetilde{\varepsilon} \to 0}
    	\begin{pmatrix}
    		M_{11} 
    		& \dfrac{1}{1!}\dfrac{\partial}{\partial \varepsilon} M_{12} 
    		& \cdots 
    		& \dfrac{1}{(N-1)!}\dfrac{\partial^{N-1}}{\partial \varepsilon^{N-1}} M_{1N} \\
    		\\
    		\dfrac{1}{1!}\dfrac{\partial}{\partial \widetilde{\varepsilon}} M_{21} 
    		& \dfrac{1}{1!}\dfrac{1}{1!}\dfrac{\partial^2}{\partial \widetilde{\varepsilon} \partial \varepsilon} M_{22} 
    		& \cdots 
    		& \dfrac{1}{1!}\dfrac{1}{(N-1)!}\dfrac{\partial^{N}}{\partial \widetilde{\varepsilon} \partial \varepsilon^{N-1}} M_{2N} \\
    		\\
    		\vdots & \vdots & \ddots & \vdots \\
    		\\
    		\dfrac{1}{(N-1)!}\dfrac{\partial^{N-1}}{\partial \widetilde{\varepsilon}^{N-1}} M_{N1} 
    		& \dfrac{1}{(N-1)!}\dfrac{1}{1!}\dfrac{\partial^{N}}{\partial \widetilde{\varepsilon}^{N-1} \partial \varepsilon} M_{N2} 
    		& \cdots 
    		& \dfrac{1}{(N-1)!}\dfrac{1}{(N-1)!}\dfrac{\partial^{2N-2}}{\partial \widetilde{\varepsilon}^{N-1} \partial \varepsilon^{N-1}} M_{NN}
    	\end{pmatrix},\\
    	\\
    	M_1& = \begin{pmatrix}
    		M & Y_1 \\ 
    		X_1 & 0 
    	\end{pmatrix}, \quad 
    	M_2 = \begin{pmatrix}
    		M & Y_2 \\ 
    		X_1 & 0 
    	\end{pmatrix}, \quad 
    	M_3 = \begin{pmatrix}
    		M & Y_3 \\ 
    		X_1 & 0 
    	\end{pmatrix},
    \end{align}
    and
    \begin{align*}
    	M_{ij} &= \frac{1}{(\zeta_1 + \varepsilon) - (\zeta_1^* + \widetilde{\varepsilon})} 
    	\varPsi_1^\dag(\zeta_1^* + \widetilde{\varepsilon}) 
    	\varPsi_1(\zeta_1 + \varepsilon), \\
    	X_1    &= \lim_{\varepsilon \to 0}
    	\begin{pmatrix}
    		\phi_1, 
    		 \partial_{\varepsilon} \phi_1,
    		 \cdots, 
    		 \partial_{\varepsilon}^{N-1} \phi_1 
    	\end{pmatrix}, \quad
    	Y_1    = \lim_{\widetilde{\varepsilon} \to 0}
    	\begin{pmatrix}
    		\varphi_1^{*},
    		 \partial_{\widetilde{\varepsilon}} \varphi_1^{*},
    		 \cdots, 
    		 \partial_{\widetilde{\varepsilon}}^{N-1} \varphi_1^{*} 
    	\end{pmatrix}^{T}, \\
    	Y_2    &= \lim_{\widetilde{\varepsilon} \to 0}
    	\begin{pmatrix}
    		\psi_1^{*}, 
    		 \partial_{\widetilde{\varepsilon}} \psi_1^{*}, 
    		 \cdots, 
    		 \partial_{\widetilde{\varepsilon}}^{N-1} \psi_1^{*} 
    	\end{pmatrix}^{T}, \quad
    	Y_3    = \lim_{\widetilde{\varepsilon} \to 0}
    	\begin{pmatrix}
    		\chi_1^{*}, 
    		 \partial_{\widetilde{\varepsilon}} \chi_1^{*}, 
    		 \cdots, 
    		 \partial_{\widetilde{\varepsilon}}^{N-1} \chi_1^{*} 
    	\end{pmatrix}^{T}.
    \end{align*}
    \begin{figure}[tbh]
    	\centering
    	\subfigure[]{\includegraphics[height=5.8cm,width=5.8cm]{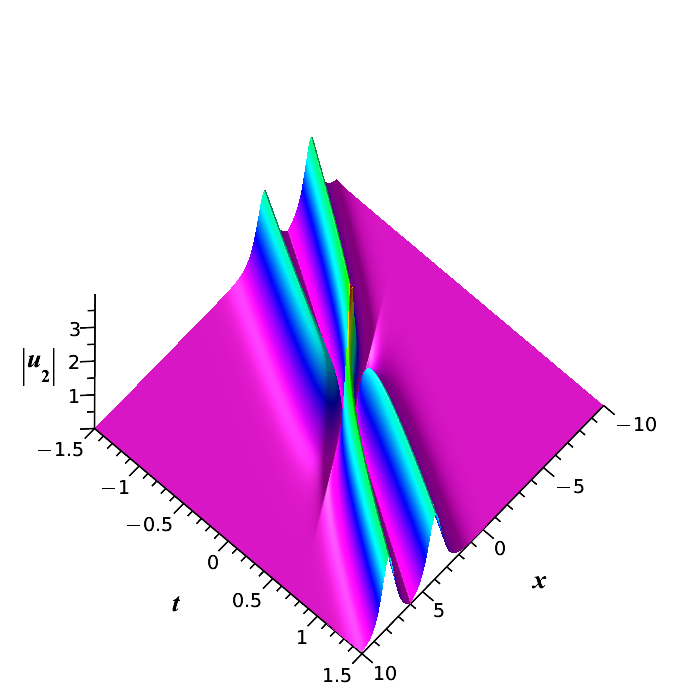}}
    	\subfigure[]{\includegraphics[height=4.8cm,width=4.8cm]{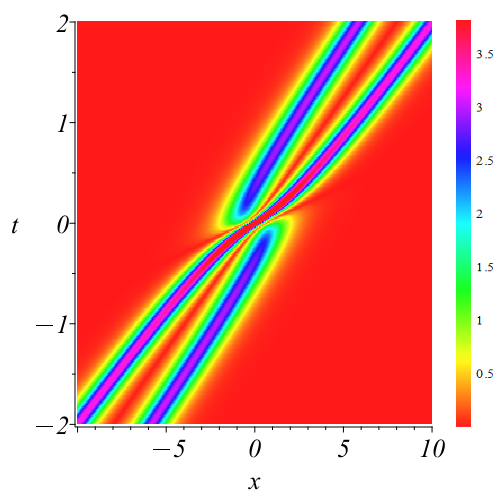}}
    	\caption{The evolution of a two-positon \eqref{pu2} with $a=1, c_1=1, c_2=\frac{1}{10}, c_3=\frac{1}{100}$ of the vector complex modified KdV equaiton on (x,t)-plane.The Panel (a) is the 3D-plot and the Panel (b) is the density plot.}
    	\label{p2}
    \end{figure} 
    When $N$=1, we can get one positon solution which is same as \eqref{c5}--\eqref{c7}. When $N$=2, we can get the two-positon as follows
    \begin{align}
    	&u[2] = -192 a c_1 \frac{
    		\frac{1}{2} A \sqrt{k} \,\mathrm{sech}\left( 8a^3 t - 2a x + \frac{1}{2} \ln k \right) - \frac{1}{24}
    	}{
    		1152 k B + \sqrt{k} \cosh\left( 16a^3 t - 4a x + \frac{1}{2} \ln k \right)
    	},\label{pu2}\\
    	&v[2] = -192 a c_2 \frac{
    		\frac{1}{2} A \sqrt{k} \,\mathrm{sech}\left( 8a^3 t - 2a x + \frac{1}{2} \ln k \right) - \frac{1}{24}
    	}{
    		1152 k B + \sqrt{k} \cosh\left( 16a^3 t - 4a x + \frac{1}{2} \ln k \right)
    	},\\
    	&w[2] = -192 a c_3 \frac{
    		\frac{1}{2} A \sqrt{k} \,\mathrm{sech}\left( 8a^3 t - 2a x + \frac{1}{2} \ln k \right) - \frac{1}{24}
    	}{
    		1152 k B + \sqrt{k} \cosh\left( 16a^3 t - 4a x + \frac{1}{2} \ln k \right)
    	}.
    \end{align}
    \begin{align}
    		A &= a^3 t - \frac{1}{12} a x - \frac{1}{24}, \quad k = c_1^2 + c_2^2 + c_3^2, \\
    		B &= a^6 t^2 - \frac{1}{6} a^4 t x + \frac{1}{144} a^2 x^2 + \frac{1}{1152}
    \end{align}
    We can see that the three components maintain the ratio $u\left[ 2 \right] :v\left[ 2 \right] :w\left[ 2 \right] =c_1:c_2:c_3$. When $N$=2,the bounded two-positon of vcmKdV equation can be ploted with $a=1, c_1=1, c_2=\frac{1}{10},c_3=\frac{1}{100}$, and the density plot of two-positon is also given, see Fig.\ref{p2}. It shows that the two postions are slowly decreasing and the shape and amplitude keeps unchanged after mutual collision. But there is a ‘‘phase shift’’ during the colliding and the amplitude reached its maximum.
    
    When $N$=3, we can get the bounded three-position with $a=1, c_1=\frac{1}{10}, c_2=\frac{1}{10}, c_3=1$, the graph of three-positon and the density plot are presented, see Fig. \ref{p3}. The three-positon also have the features similar to two-positon, the three positons keep unchanged before and after mutual collision and the amplitude reached its maximum at '$t=0$'.
    \begin{figure}[tbh]
    	\centering
    	\subfigure[]{\includegraphics[height=6.8cm,width=6.8cm]{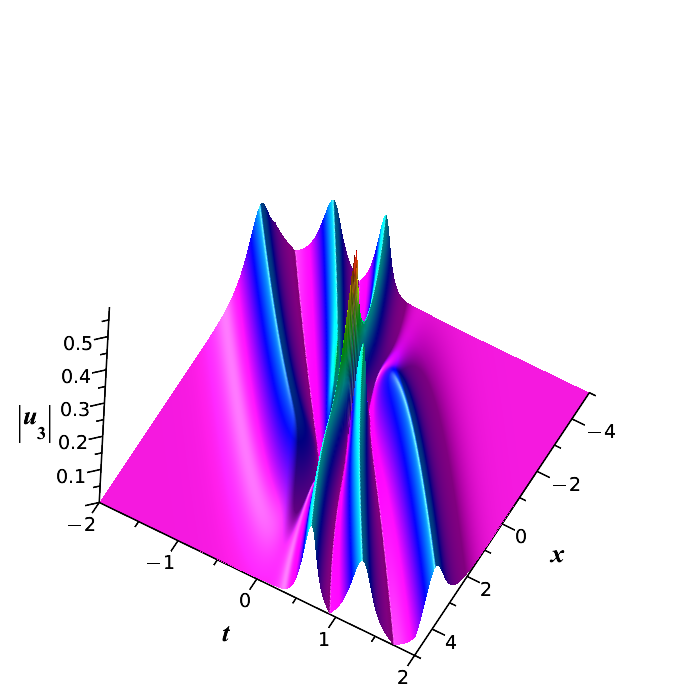}}
    	\subfigure[]{\includegraphics[height=4.8cm,width=4.8cm]{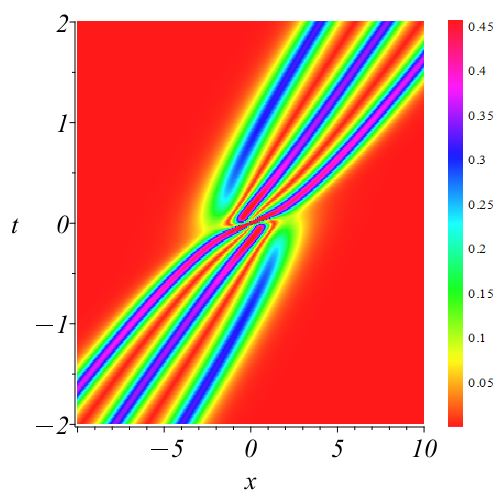}}
    	\caption{The evolution of a three-positon $|u[3]|$ with $a=1, c_1=\frac{1}{10}, c_2=\frac{1}{10}, c_3=1$ of the vector complex modified KdV equaiton on (x,t)-plane. The Panel (a) is the 3D-plot and the Panel (b) is the density plot.}
    	\label{p3}
    \end{figure} 
    
    \subsection{Rogue wave solution}
    To get the rogue wave solutions for Eq. \eqref{b1}-\eqref{b2}, we investigate the initial solution \eqref{d1} with a special eigenvalue
    \begin{equation}
    	\zeta_1(\varepsilon) = \mathrm{i} d + 2\mathrm{i} \varepsilon^2 - \dfrac{\sqrt{6}}{2} d,
    \end{equation}
    where $\varepsilon$ is a small complex parameter,we can obtain
    \begin{align}
    	\varPsi_1(x,t,\varepsilon)= 
    	\begin{pmatrix}
    		e^{\frac{\sqrt{6}}{2} \i d x} & 0 & 0 & 0 \\
    		0 & e^{-\frac{\sqrt{6}}{2} \i d x} & 0 & 0 \\
    		0 & 0 & e^{-\frac{\sqrt{6}}{2} \i d x} & 0 \\
    		0 & 0 & 0 & e^{-\frac{\sqrt{6}}{2} \i d x} \\
    	\end{pmatrix}
    	\begin{pmatrix}
    		\dfrac{
    			 B_1 e^{A_1} 
    			- B_2 e^{-A_1}
    		}{2\varepsilon} \\[3ex]
    		\dfrac{
    			a \left( e^{A_1} - e^{-A_1} \right) 
    			- 2\varepsilon \left( a_1 \alpha_2 + a_2 \alpha_1 \right) e^{A_2}
    		}{2\varepsilon} \\[3ex]
    		\dfrac{
    			a_1 \left( e^{A_1} - e^{-A_1} \right) 
    			+ 2\varepsilon a \alpha_2 e^{A_2}
    		}{2\varepsilon} \\[3ex]
    		\dfrac{
    			a_2 \left( e^{A_1} - e^{-A_1} \right) 
    			+ 2\varepsilon a \alpha_1 e^{A_2}
    		}{2\varepsilon} \\
    	\end{pmatrix},
    \end{align}
    where
    \begin{equation*}
    	A_1 = \xi_1 t + \upsilon_1 x + \sum_{i=1}^N m_i \varepsilon^{2i-1},\quad
    	A_2 = \xi_2 t + \upsilon x + \dfrac{\sqrt{6}}{2} \i d x,
    \end{equation*}
    and
    \begin{align*}
    	&\xi_1 = 12\varepsilon \sqrt{\varepsilon^2 + d} \left[ \i\sqrt{6} d \left( 2\varepsilon^2 + d \right) + \dfrac{8}{3}\varepsilon^4 + \dfrac{8}{3}\varepsilon^2 - 2d^2 \right], \\
    	&\xi_2 = -4\i \left( \i d + 2\i \varepsilon^2 - \dfrac{\sqrt{6}}{2} d \right) 
    	\left[ \i\sqrt{6} d \left( 2\varepsilon^2 + d \right) + 4\varepsilon^2 + 4d \varepsilon^2 - \dfrac{d^2}{2} \right], \\
    	&B_1=\left( 2\varepsilon \sqrt{\varepsilon^2 + d} + 2\varepsilon^2 + d \right),\quad
    	B_2=\left( -2\varepsilon \sqrt{\varepsilon^2 + d} + 2\varepsilon^2 + d \right),\\
    	&\upsilon_1 = -2\varepsilon \sqrt{\varepsilon^2 + d}, \quad
    	\upsilon = -\left( d - 2\varepsilon^2 + \dfrac{\sqrt{6}}{2} \i d \right), \quad
    	d = \sqrt{a^2 + a_1^2 + a_2^2}.
    \end{align*}
    The $\alpha_1$ and $\alpha_2$ are different real arbitrary parameters, and $m_i = m_{iR} + \mathrm{i} m_{iI} (i = 1, \dots, N)$. We expand the $\Psi_1(\varepsilon)$ as the following Taylor series
    \begin{align}
    	\Psi_1(\varepsilon) = \Psi_1^{(0)} 
    	+ \frac{1}{2!} \Psi_1^{(2)} \varepsilon^2 
    	+ \frac{1}{4!} \Psi_1^{(4)} \varepsilon^4 
    	+ \cdots 
    	+ \frac{1}{(2N-2)!} \Psi_1^{(2N-2)} \varepsilon^{2N-2} 
    	+ \mathcal{O}(\varepsilon^{2N}).
    \end{align} 
    Then we can obtain the $N$-th order rogue wave solution which has the similar form to the equation \eqref{e2}.
     \begin{equation}
    	u[N] = e^{\mathrm{i}\sqrt{6}\,a\,x}
    	\left( 
    	a + 2\mathrm{i} \frac{\det(M_1)}{\det(M)} 
    	\right), 
    	\quad
    	v[N] = e^{\mathrm{i}\sqrt{6}\,a\,x}
    	\left( 
    	a_1 + 2\mathrm{i} \frac{\det(M_2)}{\det(M)} 
    	\right), 
    	\quad
    	w[N] = e^{\mathrm{i}\sqrt{6}\,a\,x}
    	\left( 
    	a_2 + 2\mathrm{i} \frac{\det(M_3)}{\det(M)} 
    	\right),
    \end{equation}
    where
    \begin{align*}
    	M = \bigl( M_{ij} \bigr)_{1 \le i,j \le N},\quad
    	M_1 = \begin{pmatrix}
    		M & Y_1 \\
    		X_1 & 0 
    	\end{pmatrix},\quad 
    	M_2 = \begin{pmatrix}
    		M & Y_2 \\
    		X_1 & 0 
    	\end{pmatrix},\quad 
    	M_3 = \begin{pmatrix}
    		M & Y_3 \\
    		X_1 & 0 
    	\end{pmatrix},
    \end{align*}
    and
    \begin{align}
    	M^{[i,j]} = \lim_{\varepsilon \to 0 \atop \tilde{\varepsilon} \to 0}
    	\left[
    	\frac{1}{2(i-1)!} \cdot \frac{1}{2(j-1)!}
    	\cdot \frac{\partial^{2(i-1)}}{\partial \varepsilon^{2(i-1)}}
    	\cdot \frac{\partial^{2(j-1)}}{\partial \tilde{\varepsilon}^{2(j-1)}}
    	\left( 
    	\frac{
    		\Psi_1^\dagger(\tilde{\varepsilon}) \Psi_1(\varepsilon)
    	}{
    		\zeta_1(\varepsilon) - \zeta_1^*(\tilde{\varepsilon})
    	}
    	\right)
    	\right],
    \end{align}
    \begin{align*}
    	X_1 &= \lim_{\varepsilon \to 0} \exp\left( \dfrac{\sqrt{6}}{2} \i d x \right) 
    	\begin{pmatrix}
    		\phi_1,  
    		\dfrac{\partial^2}{\partial \varepsilon^2} \phi_1,  
    		\cdots,  
    		\dfrac{\partial^{2N-2}}{\partial \varepsilon^{2N-2}} \phi_1
    	\end{pmatrix}, \quad
    	Y_1 = \lim_{\widetilde{\varepsilon} \to 0} \exp\left( - \dfrac{\sqrt{6}}{2} \i d x \right) 
    	\begin{pmatrix}
    		\varphi_1^*,  
    		\dfrac{\partial^2}{\partial \widetilde{\varepsilon}^2} \varphi_1^*,  
    		\cdots,  
    		\dfrac{\partial^{2N-2}}{\partial \widetilde{\varepsilon}^{2N-2}} \varphi_1^*
    	\end{pmatrix}^{\! \! T} \\
    	Y_2 &= \lim_{\widetilde{\varepsilon} \to 0} \exp\left( - \dfrac{\sqrt{6}}{2} \i d x \right) 
    	\begin{pmatrix}
    		\psi_1^*,  
    		\dfrac{\partial^2}{\partial \widetilde{\varepsilon}^2} \psi_1^*, 
    		\cdots, &
    		\dfrac{\partial^{2N-2}}{\partial \widetilde{\varepsilon}^{2N-2}} \psi_1^*
    	\end{pmatrix}^{\! \! T}, \quad
    	Y_3 = \lim_{\widetilde{\varepsilon} \to 0} \exp\left( - \dfrac{\sqrt{6}}{2} \i d x \right) 
    	\begin{pmatrix}
    		\chi_1^*,  
    		\dfrac{\partial^2}{\partial \widetilde{\varepsilon}^2} \chi_1^*,  
    		\cdots,  
    		\dfrac{\partial^{2N-2}}{\partial \widetilde{\varepsilon}^{2N-2}} \chi_1^*
    	\end{pmatrix}^{\! \! T}.
    \end{align*}
    The first order rogue wave solution can be expressed as
    \begin{align}
    	&u\left[1\right] = e^{\i \sqrt{6} d x} \left( a \dfrac{G_1 + G_2}{G_1} + \left( a_1 \alpha_2 + a_2 \alpha_1 \right) \dfrac{G_3}{G_1} \right),\\
    	&v\left[1\right] = e^{\i \sqrt{6} d x} \left( a_1 \dfrac{G_1 + G_2}{G_1} - \alpha_2 \dfrac{G_3}{G_1} \right),\\ 
    	&w\left[1\right] = e^{\i \sqrt{6} d x} \left( a_2 \dfrac{G_1 + G_2}{G_1} - \alpha_1 \dfrac{G_3}{G_1} \right),
    \end{align}
    where
    \begin{align*}
    	G_1 &= \i\left[ -a^2\left( \alpha_{1}^{2} + \alpha_{2}^{2} \right) - \left( a_1\alpha_2 + a_2\alpha_1 \right)^2 \right] e^{-28d^3 t - 2d x} + 12d^{\dfrac{5}{2}} \left[ \sqrt{6} \left( m_{1}^{*} - m_1 \right) d^2 t + 2\i \left( m_1 + m_{1}^{*} \right) \left( d^2 t + \dfrac{x}{12} \right) \right] \\
    	&+ 12d^{\dfrac{9}{2}} \left[ \sqrt{6} \left( m_{1}^{*} - m_1 \right) + 2\i \left( m_1 + m_{1}^{*} \right) \right] t + 2\i \left( m_1 + m_{1}^{*} \right) d^{\dfrac{3}{2}} + 2\i d^{\dfrac{5}{2}} \left( m_1 + m_{1}^{*} \right) x \\
    	&+ \i a^2 \left( -1440 d^5 t^2 - 96 d^3 t x - 4 d x^2 - m_1 m_{1}^{*} \right) - 144\i d^2 \left( d^2 + a_{1}^{2} + a_{2}^{2} \right) t^2 - 96\i d^3 \left[ \left( d^2 + a_{1}^{2} + a_{2}^{2} \right) x + d \right] t \\ 
    	&- 4\i d \left[ 1 + \left( d^2 + a_{1}^{2} + a_{2}^{2} \right) x^2 + 2d x \right],\\
    	G_2 &= -96d^{\dfrac{9}{2}} \left[ \dfrac{\sqrt{6}}{2} \left( m_{1}^{*} - m_1 \right) + \i \left( m_1 + m_{1}^{*} \right) \right] t - 8\i d^{\dfrac{5}{2}} \left( m_1 + m_{1}^{*} \right) x - 96 \sqrt{6} d^4 t \\
    	&+ 5760\i d^2 \left[ d^5 t^2 + \dfrac{1}{15} \left( d^3 x + \dfrac{1}{2} d^2 \right) t + \dfrac{1}{360} d x^2 + \dfrac{1}{1440} m_1 m_{1}^{*} + \dfrac{1}{360} x \right],\\
    	G_3& = 8 \left[ \i d^{\dfrac{5}{2}} x + 6 \left( 2\i + \sqrt{6} \right) d^{\dfrac{9}{2}} - \dfrac{\i m_1 d^2}{2} + \i d^{\dfrac{3}{2}} \right] e^{-\left( 14 + 3 \sqrt{6}  \i \right) d^3 t - d x}.
    \end{align*}
    \begin{figure}[tbh]
    	\centering
    	\subfigure[]{\includegraphics[height=4.8cm,width=4.8cm]{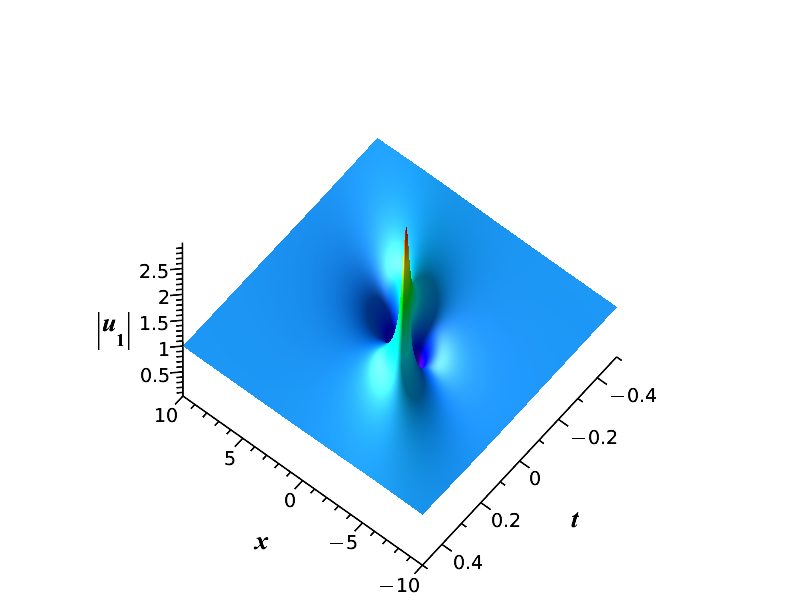}}
    	\subfigure[]{\includegraphics[height=4.8cm,width=4.8cm]{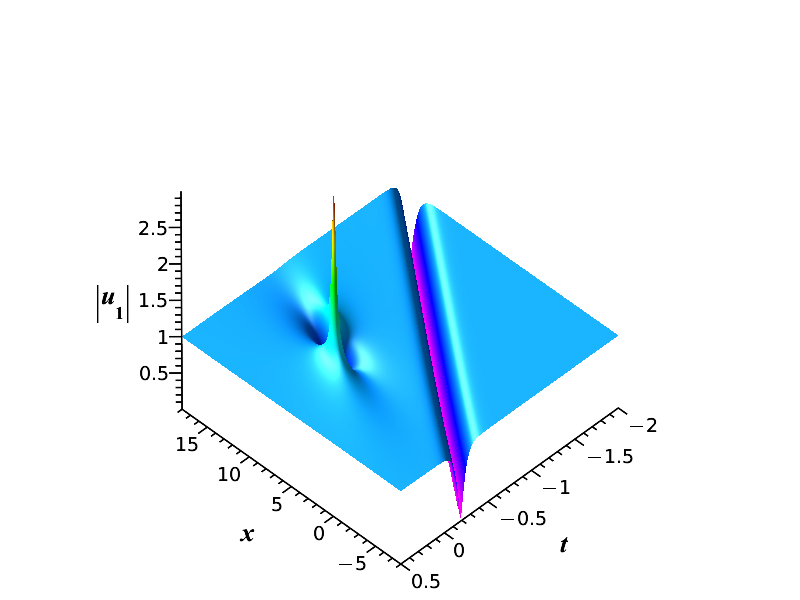}}
    	\subfigure[]{\includegraphics[height=4.8cm,width=4.8cm]{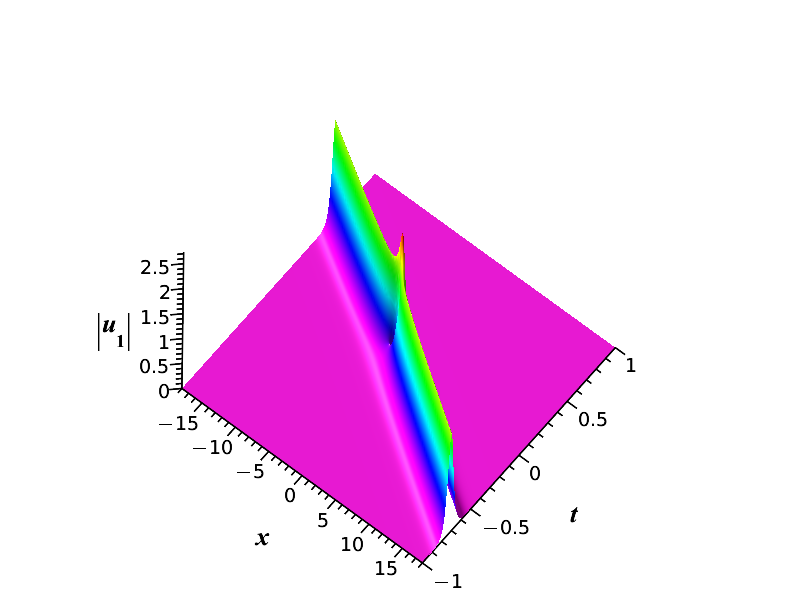}}
    	\caption{Panel (a) is the Peregrine rogue wave with $\alpha_1=0, \alpha_2=0, a_1=0, a_2=0, a=1, m_1=0$.Panel (b) is a bounded rogue wave coexist with a dark soliton with $\alpha_1=\frac{1}{1000}, \alpha_2=0, a_1=0, a_2=0, a=1, m_1=10+10\i$.Panel (c) is the bounded rogue wave interact with a bright soliton with $\alpha_1=10, \alpha_2=0, a_1=1, a_2=1, a=0, m_1=0$.}
    	\label{r1}
    \end{figure} 
    \begin{figure}[tbh]
    	\centering
    	\subfigure[]{\includegraphics[height=6.8cm,width=6.8cm]{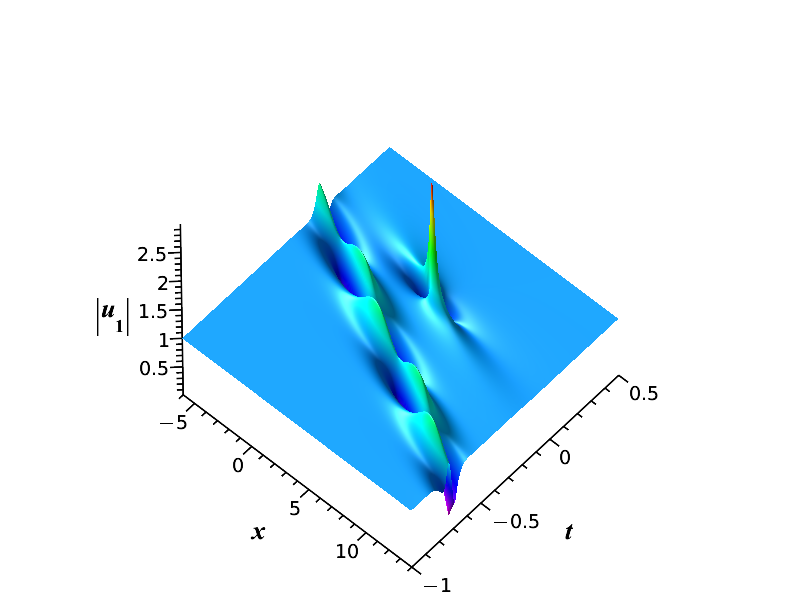}}
    	\subfigure[]{\includegraphics[height=4.8cm,width=4.8cm]{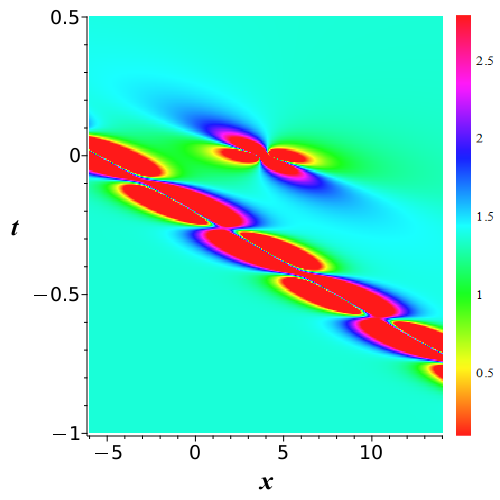}}
    	\caption{Panel(a) is the a bounded rogue wave coexist with a breather-breather wave with $\alpha_1=\frac{1}{100}, \alpha_2=\frac{1}{100}, a_1=1, a_2=0, a=1, m_1=10$.Panel (b) is the density plot of (a).}
    	\label{r2}
    \end{figure} 
    The dynamics of the one-order rogue wave can also be analyzed into the different cases similar to the breather solutions.
    
    (1)If $\alpha_1=0, \alpha_2=0$, the three components have the ratio $u\left[ 1 \right] :v\left[ 1 \right] :w\left[ 1 \right] =c_1:c_2:c_3$, and the standard Peregrine rogue wave can be presented, see Fig. \ref{r1} (a).
    
    (2)If $\alpha_1 \ne 0$ and $\alpha_2 \ne 0$ are not both zero, we can obtain the bounded dark-bright solitons coexist with rogue wave,see Fig.\ref{r1}. If $a \ne 0, \left( a_1\alpha _2+a_2\alpha _1 \right)=0$, we can get the bounded rogue wave coexist with a dark soliton. If $a = 0, \left( a_1\alpha _2+a_2\alpha _1 \right) \ne 0$, we can get the bounded rogue wave interact with a bright soliton.
    
    (3)If $\alpha_1 \ne 0$ and $\alpha_2 \ne 0$ are not both zero, $a \ne 0$ , $\left( a_1\alpha _2+a_2\alpha _1 \right) \ne 0$ and with $p_{1R}=0$, we can obtain a rogue wave coexist with a breather-breather solution, see Fig.\ref{r2}.

    \section{Conclusions}\label{5}
    In this paper, we have investigated the determinant representation of the $N$-fold Darboux transformation for the vector complex modified KdV equation. Using this framework, we have derived $N$-bright-bright-bright soliton solutions, $N$-dark-bright-bright soliton solutions, $N$-breather solutions, $N$-positon solutions, and $N$th-order rogue wave solutions. With the exception of the Akhmediev breather, which can develop singularities, all derived solutions are globally bounded. The dynamics of the soliton solutions are presented and we give the asymptotic analysis for bounded bright-bright-bright soliton. It is worth to noting that there is a periodic oscillatory waves coexist or interact with the soliton solutions. The bounded Akhmediev breather, collapsed Akhmediev breather and the regular breather, alone with the breather interact with dark-bright solitons, breathers and soliton-breather-breather waves are shown. The positons, rogue wave and rogue wave coexist or interact with dark-bright solitons and breather-breather waves are exhibited.

    \section*{Acknowledgments}
    {This work is supported by the Zhejiang Provincial Natural Science Foundation of China under Grant No.LY24A010002, the National Natural Science Foundation of
    	China under Grant No. 12171433, the Natural Science Foundation of Ningbo under Grant No. 2023J126.}\\
    
    \noindent\textbf{Compliance with ethical standards}\\
    
    \noindent \textbf{Ethical Statement}\\ {\small Authors declare that they comply with ethical standards. } \\
    
    \noindent \textbf{Conflict of interest}\\ {\small Authors declare that they have no conflict of interest.} \\
    
    \noindent\textbf{Data Availability Statement}\\
    {\small The data that support the findings of this article are available from the corresponding author, upon reasonable request. } \\
    

\vspace{-0.5cm}
    
    
    
    \begin{thebibliography}{}
    	%
    	%
    	\bibitem{Clifford1967} 
    	Clifford S. Gardner, John M. Greene, Martin D. Kruskal, et al. Method for Solving the Korteweg-de Vries Equation. Phys. Rev. Lett. 19(1967) 1095-1097.
    	
    	\bibitem{Ablositz1981}
    	M. J. Ablowitz and H. Segur. Solitons and the Inverse Scattering Transform. SIAM Philadelphia (1981).
    	
    	\bibitem{Pan1998}  
    	Pan Z., Zheng K., Zhao S. An algebraic method for solving the KdV equation (I). Two-parameter solution family. Chaos. Soliton. Fract. 9(1998) 1733-1737.
    	
    	\bibitem{Dejak2006} 
    	Dejak, SI, Sigal, IM. Long-time dynamics of KdV solitary waves over a variable bottom. Commun Pur. Appl. Math. 59(2006) 869-905.
    	
    	\bibitem{Stephen2011}
    	Stephen C. Anco, Nestor Tchegoum Ngatat, Mark Willoughby. Interaction properties of complex mKdV solitons. Physica D 240(2011) 1378-1394.
    	
    	
    	\bibitem{Das2024}
    	Dulal Chandra Das, Samiran Das, Rekha Kalita. Modified Korteweg-de Vries solitons with quartic nonlinearity in a dusty plasma. Phys. Scripta. 99(2024) 055266.
    	
    	\bibitem{Schief1995}
    	Schief, Wk. An Infinite Hierarchy of symmetries associated with hyperbolic surfaces. Nonlinearity 8(1995) 1-9.
    	
    	\bibitem{GESZTESY1991}
    	Gesztesy F, Schweiger W, Simon B. Commutation Methods Applied to the mKdV-equation. T.Am.Math.Soc. 324(1991) 465-525.
    	
    	\bibitem{Nagatani1998}
    	Nagatani T. Modified KdV equation for jamming transition in the continuum models of traffic. Physica A. 261(1998) 599-607.
    	
    	\bibitem{He2005}
    	He J., Chen S. Hamiltonian Formalism of mKdV Equation with Non-vanishing Boundary Values. Commun. Theor. Phys. 44(2005) 321-325.
    	
    	\bibitem{Salas2010}
    	Salas, Alvaro H. Exact solutions to mKdV equation with variable coefficients. Appl.Math.Comput. 216(2010) 2792-2798.
    	
    	
    	
    	\bibitem{Wu2017}
    	Wu J., Geng X. Inverse Scattering Transform and Soliton Classification of the Coupled Modified Korteweg-de Vries Equation. Commun.  Nonlinear. Sci. 53(2017) 83-93.
    	
    	\bibitem{Zhang2020}
    	Zhang G., Yan Z. Focusing and defocusing mKdV equations with nonzero boundary conditions: Inverse scattering transforms and soliton interactions. Physica D 410(2020) 132521.
    	
    	\bibitem{Liu2021}
    	Liu N., Guo B. Asymptotics of solutions to a fifth-order modified Korteweg-de Vries equation in the quarter plane. Anal. Appl. 19(2021) 575-620.
    	
    	\bibitem{Gao2021}
    	Gao X., Guo Y. Shan W., et al. In the Atmosphere and Oceanic Fluids: Scaling Transformations, Bilinear Forms, Bäcklund Transformations and Solitons for A Generalized Variable-Coefficient Korteweg-de Vries-Modified Korteweg-de Vries Equation. China Ocean Eng. 35(2021) 518-530.
    	
    	\bibitem{Roy2022}
    	Roy, Subrata, Raut, Santanu, Kairi, Rishi Raj. Nonlinear analysis of the ion-acoustic solitary and shock wave solutions for non-extensive dusty plasma in the framework of modified Korteweg-de Vries-Burgers equation. Pramana-J.Phys. 96(2022) 1-13.
    	
    	\bibitem{Chen2023}
    	Chen M., Fan E., He J. Riemann-Hilbert approach and the soliton solutions of the discrete mKdV equations. Chaos. Soliton. Fract. 168(2023) 113209.
    	
    	\bibitem{Vladimir2023}
    	Vladimir I. Kruglov, Houria Triki. Interacting Solitons, Periodic Waves and Breather for Modified Korteweg–de Vries Equation. Chinese Phys. Lett. 40(2023) 19-23.
    	
    	\bibitem{Liu2024}
    	Liu X., Lu B., Zhang, Da-Jun. New type of solutions for the Modified Korteweg-de Vries equation. Appl. Math. Lett. 159(2025) 109288.
    	
    	\bibitem{Muhammad2024}
    	Muhammad, Jan, Younas, Usman, Hussain, Ejaz, et al. Analysis of fractional solitary wave propagation with parametric effects and qualitative analysis of the modified Korteweg-de Vries-Kadomtsev-Petviashvili equation. Sci. Rep-UK. 14(2024) 1-16.
    	
    	
    	
    	\bibitem{Zhaqilao2013}
    	Zha Q. Nth-Order Rogue Wave Solutions of the Complex Modified Korteweg-de Vries Equation. Phys. Scripta 87(2013) 065401.
    	
    	\bibitem{Huang2021}
    	Huang L., Lv N. Soliton Molecules, Rational Positons and Rogue Waves for the Extended Complex Modified KdV Equation. Nonlinear Dynam. 105(2021) 3475-3487.
    	
    	
    	\bibitem{Ismail2008}
    	Ismail M. S. Numerical solution of complex modified Korteweg-de Vries equation by Petrov-Galerkin method. Appl. Math. Comput. 202(2008) 520-531.
    	
    	\bibitem{Uddin2009}
    	Uddin, Marjan, Haq, Sirajul, Siraj-ul-Islam. Numerical solution of complex modified Korteweg-de Vries equation by mesh-free collocation method. Comput.Math.Appl. 58(2009) 566-578.
    	
    	\bibitem{Tao2021}
    	Xu T., Zhang G., Wang L., et al. Numerical Simulation of the Soliton Solutions for a Complex Modified Korteweg-de Vries Equation by a Finite Difference Method. Commun. Theor. Phys. 73(2021) 41-51.
    	
    	\bibitem{Yuan2023}
    	Yuan F. The semi-rational solutions of the (2+1)-dimensional cmKdV equations. Nonlinear.Dynam. 111(2023) 733-744.
    	
    	\bibitem{Wangwazwaz2022}
    	Wang G., Wazwaz, Abdul-Majid. A New (3+1)-Dimensional KdV Equation and mKdV Equation With Their Corresponding Fractional Forms. Fractals 30(2022) 2250081.
    	
    	\bibitem{Bai2024}
    	Bai Q., Li X., Zhao Q. Evolution of Dispersive Shock Waves to the Complex Modified Korteweg-de Vries Equation With Higher-Order Effects. Chaos Soliton Fract. 182(2024) 114731.
    	
    	\bibitem{Rao2024}
    	Rao, J., Mihalache, Dumitru, He J. Multiple Solitons and Breathers on Periodic Backgrounds in the Complex Modified Korteweg-de Vries Equation. Appl. Math. Lett. 160(2025) 109308.
    	
    	\bibitem{Zhao2024}
    	Zhao Y., Zhu D. A Riemann-Hilbert Approach for the Focusing and Defocusing mKdV Equation With Asymmetric Boundary Conditions in Few-Cycle Pulses. Eur. Phys. J. Plus. 139(2024) 603.
    	
    	\bibitem{Song2024}
    	Song C., Liu D., Ma L. Soliton Solutions of a Novel Nonlocal Hirota System and a Nonlocal Complex Modified Korteweg-de Vries Equation. Chaos Soliton Fract. 181(2024) 114707.
    	
    	\bibitem{Xu2025}
    	Xu B., Zhang S. Exact Solutions of a Local Fractional Nonisospectral Complex mKdV Equation Based on Riemann-Hilbert Method With Time-Varying Spectrum. Alex. Eng. J. 115(2025) 564-576.
    	
    	
    	
    	\bibitem{Liu2016}
    	Liu H., Geng X. Initial-Boundary Problems for the Vector Modified Korteweg-de Vries Equation via Fokas Unified Transform Method. J. Math. Anal. Appl. 440(2016) 578-596.
    	
    	\bibitem{Wang2020}
    	Wang X., Han B., Application of the Riemann-Hilbert Method to the Vector Modified Korteweg-de Vries Equation. Nonlinear Dynam.  99(2020) 1363-1377.
    	
    	\bibitem{Fenchenko2018}
    	Fenchenko, V. Khruslov, E. Nonlinear Dynamics of Solitons for the Vector Modified Korteweg-de Vries Equation. J. Math. Phys. Anal. Geo. 14(2018) 153-168.
    	
    	\bibitem{Matveev1992}
    	Matveev, V. B. Positon positon and soliton positon Collisions-KdV case. Phys. Lett. A  166(1992) 209-212.
    	
    	\bibitem{Stahlhofen1992}
    	Stahlhofen A.A. Positions of the Modified Korteweg-de Vries Equation. Ann. Phys-Berlin. 1(1992) 554-569.
    	
    	\bibitem{Beutler}
    	Beutler, R. Positon Solutions of the Sine-Gordon Equation. J. Math. Phys. 34(1993) 3098-3109.
    	
    	\bibitem{Hu2021}
    	Hu A., Li M., He J. Dynamic of the Smooth Positons of the Higher-Order Chen-Lee-Liu Equation. Nonlinear Dynam. 104(2021) 4329-4338.
    	
    	\bibitem{Shan2024}
    	Shan J., Li M. The Dynamic of the Positons for the Reverse Space-Time Nonlocal Short Pulse Equation. Physica D 470(2024) 134419.
    	
    	\bibitem{Rahman2025}
    	Rahman, Riaz Ur, Li Z., He J. Magnetic Wave Dynamics in Ferromagnetic Thin Films: Interactions of Solitons and Positons in Landau-Lifshitz-Gilbert Equation. Physica D  479(2025) 134719.
    	
    	\bibitem{Matveev2002}
    	Matveev, V. B. Positons: Slowly Decreasing Analogues of Solitons. Adv. Theor. Math. Phys. 131(2002) 483-497.
    	
    	\bibitem{Akhmediev2009}
    	Akhmediev, N. Soto-Crespo, J. M. Ankiewicz A. Extreme waves that appear from nowhere: On the nature of rogue waves.Phys.Lett.A. 373(2009) 2137-2145.
    	
    	\bibitem{Akhmediev2023}
    	Akhmediev, Nail. Waves That Appear From Nowhere. Proc. Roy. Soc. Victoria 135(2023) 64-68.
    	
    	\bibitem{Kaur2022}
    	Kaur, Rajneet, Singh, Kuldeep, Saini, N. S. Electron Acoustic Rogue Waves in Earth's Magnetosphere. J. Astrophys. Astr. 43(2022) 62.
    	
    	\bibitem{Da2020}
    	Da Silva Mendes, Saulo Matusalem. On the Statistics of Oceanic Rogue Waves in Finite Depth: Exceeding Probabilities, Physical Constraints and Extreme Value Theory. The University Of North Carolina At Chapel Hill  2020.
    	
    	\bibitem{Qi2025}
    	Qi J., Wang D. High-Order Optical Rogue Waves in Two Coherently Coupled Nonlinear Schrödinger Equations. Physica D 472(2025) 134538.
    	
    	\bibitem{Kengne2023}
    	Kengne, Emmanuel. Manipulating Matter Rogue Waves in Bose-Einstein Condensates Trapped in Time-Dependent Complicated Potentials. Nonlinear Dynam. 111(2023) 11497-11520.
    	
    	\bibitem{Haefner2023}
    	Haefner, Dion, Gemmrich, Johannes, Jochum, Markus. Machine-Guided Discovery of a Real-World Rogue Wave Model. P. Natl. Acad. Sci. USA. 120(2023) e2306275120.
    	
    	\bibitem{Efimov2010}
    	Efimov, V. B. Ganshin, A. N. Kolmakov, G. V., et al. Rogue Waves in Superfluid Helium. Eur. Phys. J-Spec. Top. 185(2010) 181-193.
    	
    	\bibitem{Zhen2023}
    	Zhen Y., Rogue Waves on the Periodic Background in the Complex Modified KdV Equation With Higher-Order Effects. Wave Motion 123(2023) 103209.
    	
    	\bibitem{Li2020}
    	Li R., Geng X. Rogue Periodic Waves of the Sine-Gordon Equation. Appl. Math. Lett. 102(2020) 106147.
    	
    	\bibitem{Ying2021}
    	Yang Y., Gao Y., Yang H. Analysis of the Rogue Waves in the Blood Based on the High-Order NLS Equations With Variable Coefficients. Chinese Phys. B 30(2021) 110202.
    	
    	\bibitem{Wang2024}
    	Wang C., Wang L., Li C. Rogue Waves and Their Dynamics in the Ito's System With the Nonzero Constant Background. Nonlinear Dynam. 112(2024) 6547-6559.
    	
    	\bibitem{Terng2000}
    	Terng, C. L. Uhlenbeck, K. Bäcklund.  Transformations and loop group actions. Commun. Pur. Appl. Math. 53(2000) 1-75.
    	
    	\bibitem{Wang2022}
    	Wang X., Li C. Solitons, breathers and rogue waves in the coupled nonlocal reverse-time nonlinear Schrödinger equations. J. Geom. Phys. 180(2022) 104619.
    \end{thebibliography}
\end{document}